\documentclass[conference]{IEEEtran}
\usepackage{cite}
\ifCLASSINFOpdf
   \usepackage[pdftex]{graphicx}
   \DeclareGraphicsExtensions{.pdf,.jpeg,.png}
\else
   \usepackage[dvips]{graphicx}
   \DeclareGraphicsExtensions{.eps}
\fi
\usepackage[cmex10]{amsmath}
\usepackage{bbm,amsthm}
\newtheorem{theorem}{Theorem}
\newtheorem{claim}{Claim}

\newtheorem{lemma}{Lemma}

\newtheorem{remark}{Remark}
\newtheorem{example}{Example}

\newcommand{\beqno}{ \begin{equation*} }
\newcommand{\eeqno}{ \end{equation*} }
\newcommand{\beq}{ \begin{equation} }
\newcommand{\eeq}{ \end{equation} }

\begin{document}
%\title{A class of capacity-achieving sum-networks}
\title{Sum-networks from undirected graphs: construction and capacity analysis}

\author{\IEEEauthorblockN{Ardhendu Tripathy and Aditya Ramamoorthy}
\IEEEauthorblockA{Department of Electrical and Computer Engineering,
Iowa State University, Ames, Iowa 50011\\
Email: \{ardhendu,adityar\}@iastate.edu}
}


%Telephone: (800) 555--1212, Fax: (888) 555--1212}


\maketitle

\begin{abstract}
%\boldmath
We consider a directed acyclic network with multiple sources and multiple terminals where each terminal is interested in decoding the sum of independent sources generated at the source nodes. We describe a procedure whereby a simple undirected graph can be used to construct such a \textit{sum-network} and demonstrate an upper bound on its computation rate. Furthermore, we show sufficient conditions for the construction of a linear network code that achieves this upper bound. Our procedure allows us to construct sum-networks that have any arbitrary computation rate $\frac{p}{q}$ (where $p,q$ are non-negative integers). Our work significantly generalizes a previous approach for constructing sum-networks with arbitrary capacities. Specifically, we answer an open question in prior work by demonstrating sum-networks with significantly fewer number of sources and terminals.
\end{abstract}


% \fi
\IEEEpeerreviewmaketitle

\section{Introduction}

Function computation using network coding is an area that has received significant attention in recent years (see for instance \cite{appuFKZ11, appuFKZ13, ramamoorthyL13, raiD12}). Broadly speaking, one considers directed acyclic networks with error free links, a set of source nodes that generate independent information and terminal nodes that demand a certain function of the source values. The topology of the network is allowed to be arbitrary. The most general formulation is evidently quite complex to study as depending on the function the demands can be arbitrarily complex and contain for instance multiple unicast as a special case, a problem which is well recognized to be hard (see for instance the discussion in \cite{huangR14, huangR12_TCOM}). Accordingly, several simplified settings have been considered in the literature. The work of \cite{appuFKZ11, appuFKZ13} considers general functions, but networks with only one terminal. A different line of work considers networks with multiple terminals that each need a simple function such as the sum \cite{ramamoorthyL13, raiD12}. Significant prior work \cite{kornerM79, orlitskyR01, doshiSMJ07} considers information theoretic issues in function computation where the sources are dependent but the network structures are simple in the sense that there are direct links between sources and the terminals.

In this work, we consider the problem of network coding for sum-networks. Specifically, the problem is one of multicasting the finite field sum of source values that are available at a set of source nodes to a designated set of terminals over a directed acyclic network, i.e., a sum-network. The topology of the graph denoting the network can be completely arbitrary under the trivial restriction that there exists a path between any terminal and each source node. We assume that the sources are independent and that the network links are delay and error-free but have finite capacity.

This problem was first considered in the work of \cite{ramamoorthy08}, where the class of networks with unit-entropy sources, unit-capacity edges and either two sources or two terminals was considered. For this class of networks it was demonstrated that as long as each source is connected to each terminals, computation of the sum was possible. In contrast, it was shown \cite{ramamoorthyL13} that there exist networks with three sources and three terminals where sum computation is impossible even though each source terminal pair is connected. Conditions on sum computation for networks with three sources and three terminals have also been investigated in prior work \cite{ramamoorthyL13,SD10}. More generally, one can define the notion of computation rate \cite{appuFKZ11}. Informally, a network code for a sum-network is said to have rate $r/l$ if in $l$ time slots, one can multicast the sum $r$ times to all the terminals. A network is called solvable if it has a $(r,r)$ code and not solvable otherwise. The problem of multicasting the sum can be shown to be equivalent to the problem of multiple unicast. Specifically, \cite{raiD12} shows that there is a linearly solvable equivalent sum-network for any multiple-unicast network. Thus characterizing the solvability of sum networks and identifying the network resources required in order to ensure solvability of a sum-network assumes importance.
%Rai et al \cite{rai10} collate some results on solvability and capacity of sum-networks.

There are several open problems for sum-networks where the number of sources and terminals is at least three. For instance, nontrivial sufficient conditions on the network resources that allow for sum computation are not known. However, recently certain impossibility results have been obtained. Reference \cite{rai13capacity} shows that for any integer $k \geq 2$, there exists a sum-network with three sources and four or more terminals (and also a sum-network with three terminals and four or more sources) with coding capacity $\frac{k}{k+1}$ for integers $k \geq 0$. Reference \cite{raiD13} is most closely related to our work. Given a ratio $p/q$ it constructs a sum-network that has capacity equal to $p/q$. In this work, we propose a construction of sum networks that significantly generalizes the work of \cite{raiD13} and answer some of its open questions.
\subsection{Main Contributions}
%\begin{itemize}
%\item
\noindent $\bullet~$Asssuming coprime $p$ and $q$, the work of \cite{raiD13}, constructs a sum-network that has $2q-1 + \binom{2q-1}{2}$ sources and $2q-1 + \binom{2q-1}{2}+1$ terminals, which can be significantly large if $q$ is large. An open question posed in \cite{raiD13} is whether sum-networks of smaller size exist that have a capacity $p/q$. In this work, we answer this in the affirmative. Specifically, we construct a large family of sum-networks that can be significantly smaller for several values of $p/q$ and recover their result as a special case. We note that examples of small sum-networks with capacity strictly smaller than one are useful in investigating sufficiency conditions for general networks. For example, \cite{ramamoorthyL13} demonstrates an instance of a network with three sources and three terminals where unit connectivity between each source terminal pair does not suffice, implying that a sufficiency condition for such 3-source, 3-terminal networks that only looks at minimum connectivity between source terminal pairs needs to consider networks where the minimum cut between each source terminal pair is at least two.\\
%\item
\noindent $\bullet~$ Our proof that the constructed sum-network has the appropriate capacity value is simpler than the proof of \cite{raiD13}.\\
%\end{itemize}
%In this paper we describe a class of sum-networks that are not solvable (i.e. they do not admit a unit rate network code). We give an upper bound on the rate of a network code that computes the sum of sources at every terminal and show that this upper bound is tight. Some results in this direction have been obtained by Rai and Das \cite{raiD13}. We show that their network constructions are a special case of our constructions. In addition, we answer some of the questions they raised in their paper.
This paper is organized as follows. The problem is formally posed in Section \ref{sec:setup}, our construction is explained in Section \ref{sec:constr} and a comparison with existing results appears in Section \ref{sec:comp}. We conclude the paper with a discussion about future work in Section \ref{sec:conclusions}.
\section{Problem formulation}
\label{sec:setup}
We consider communication over a directed acyclic graph (DAG) $G=(V,E)$ where $V$ is the  set of nodes and $E \in V \times V \times \mathbf{Z}_+$ are the edges denoting the delay-free communication links between them.
Subset $S \subset V$ denotes the source nodes and $T \subset V$ denotes the terminal nodes. The source nodes have no incoming edges and the terminal nodes have no outgoing edges. Each source node $s_i \in S$ generates an independent random process $X_i$, such that the sequence of random variables $X_{i1}, X_{i2}, \dots$ indexed by time are i.i.d. and each $X_{ij}$ takes values that are uniformly distributed over a finite alphabet $\mathcal{A}$ that is assumed to be a finite field such that $|\mathcal{A}| = q$. Each edge is of unit capacity and can transmit one symbol from $\mathcal{A}$ per unit time. Our model allows for multiple edges between nodes. In this case the edges are given an additional index. For instance if there are two edges between nodes $u$ and $v$, these will be represented as $(u,v,1)$ and $(u,v,2)$. The capacity of the edge $(u,v)$ is defined as the number of edges between $u$ and $v$.

We use the notation $\text{In}(v)$ to represent the set of incoming edges at node $v \in V$. We will also work with undirected graphs in this paper. If $v$ is a node in an undirected graph, then $\text{In}(v)$ will represent the edges incident on $v$.
%
% Each source node $s_i \in S$ generates an independent random process $X_i$. Without loss of generality, we can assume the sources to be unit-entropy sources as each random process is independent and a non-unit entropy source can be represented as several unit-entropy rate sources. The source nodes have no incoming edges and the terminal nodes have no outgoing edges. The sequence of random variables $X_{i1}, X_{i2}, \dots$ is an i.i.d. sequence and each $X_{ij}$ takes values in a fixed alphabet $\mathcal{A}$ which we assume to be a finite field. %For our code constructions, we embed the alphabet $\mathcal{A}$ into a finite field of size $q \geq |\mathcal{A}|$.

An network code is an assignment of edge functions to each edge in $E$ and a decoding function to each terminal in $T$. The edge function for an edge connected to a source, depends only the source values. Likewise an edge function for an edge that is not connected to a sources depends on the values received on its incoming edges and the decoding function for a terminal depends only on its incoming edges. %A $(r,l)$ fractional code is described by the following.
We let the source messages be vectors of length $r$ and the edge functions to be vectors of length $l$. The decoding functions should be such that each terminal recovers the sum of all the source message vectors. The domain and range of the encoding functions can be summarized as follows.
\begin{itemize}
\item Edge function for edge $e$. %An edge function for edge $e$ is defined as
\begin{align*}
\phi_e &: \mathcal{A}^r \rightarrow \mathcal{A}^l ~~ \text{if tail}(e)\in S,  \\
\phi_e &: \mathcal{A}^{l|\text{In}(\text{tail}(e))|} \rightarrow \mathcal{A}^l ~~ \text{if tail}(e) \notin S.
\end{align*}

\item Decoding function for the terminal $t_i \in T$.
\begin{equation*}
\psi_{t_i} : \mathcal{A}^{l|\text{In}(t_i)|} \rightarrow \mathcal{A}^r
\end{equation*}
\end{itemize}

A network code is a linear network code if all the edge and decoding functions are linear. For the sum-networks that we consider, a $(r,l)$ fractional network code solution over $\mathcal{A}$ is such that the sum of $r$ source symbols (over the finite field) can be communicated to all the terminals in $l$ units of time. The rate of this network code is defined to be $r/l$. A network is said to be solvable if it has a $(r,r)$ network coding solution for some $r \geq 1$. A network is said to have a scalar solution if it has a $(1,1)$ solution. The supremum of all achievable rates is called the capacity of the network.

\section{Construction and capacity of a family of sum networks}
\label{sec:constr}
In this section we construct a family of sum-networks which are not solvable even though each source terminal pair is connected via at least one path. In fact, we will demonstrate that there exist families of sum-networks that are not solvable even though each source terminal pair has a minimum cut that is strictly larger than a fixed constant. The construction of the sum-network starts with a parallel set of $b \geq 2$ edges that we refer to as bottleneck edges in the subsequent discussion. Each bottleneck edge has a capacity $\alpha$. %$s_i \in S~\forall i $ is the set of source nodes and $ T $ is the set of terminal nodes. The objective is to find a network code which is able to communicate the sum of the source symbols to all the terminals.
%In order to explain the construction we use refer to a simple undirected graph $\widetilde{G} = (\widetilde{V}, \widetilde{E})$ whose vertices and edges allow us to identify the corresponding sum-network.
%The set $S$ consists of source nodes in the DAG $G$ and $T$ consists of terminal nodes.
There are $\alpha$-capacity edges connecting carefully chosen subsets of the source nodes $S$ to the tail of each bottleneck edge. %The subsets of sources connected likewise to the bottleneck edges are in general different for different bottlenecks.
Similarly, $\alpha$-capacity edges connect the head of each bottleneck edge to a subset of terminal nodes $T$.
The choice of the various subsets of the source nodes is made via the help of
%In order to describe the various subsets connecting to different bottleneck edges, we take help of
a undirected simple connected graph\footnote{A graph is said to be simple if does not have self loops and multiple edges between a pair of nodes.} $\widetilde{G}=(\widetilde{V},\widetilde{E})$ where $|\widetilde{V}|=b$.
We arbitrarily number the vertices in $\widetilde{G}$ as $1,2, \hdots , |\widetilde{V}|$. Then there are potentially the following three sets of sources in the sum-network $G$.
\begin{enumerate}
\item $S_1 = \{s_i : i = 1,2, \hdots, |\widetilde{V}|\}$,
\item $S_2 = \{s_e : e \in \widetilde{E}\}$, and
\item $s_ \star$.
\end{enumerate}
Similarly, there are three types of terminals in $G$.
\begin{enumerate}
\item $T_1=\{t_i : i=1,2, \hdots , |\widetilde{V}|\}$,
\item $T_2=\{t_e : e \in \widetilde{E}\}$, and
\item $t_ \star$,
\end{enumerate}
so that $T = T_1 \cup T_2 \cup \{t^\star\}$.
We propose two different constructions depending on whether we include $s_\star$ in the sum-network or not. %The analysis for both these networks is quite similar in nature. In our discussion, we discuss networks where $s^*$ exists and simply state the results that are obtained when it is absent. %There are independent random processes generated at each source node.
In the constructions, we assume that $\widetilde{G}$ is connected and that $|\widetilde{E}| \geq |\widetilde{V}|$, i.e., it is not a tree.
\subsection{Construction 1}
For the first construction we do not include $s_\star$ in the network, thus $S = S_1 \cup S_2$. Let
\begin{equation}
\label{eq:sources}
X=\{X_1, X_2, \hdots , X_{|\widetilde{V}|}\} \cup \{X_e : e \in \widetilde{E}\}% \cup \{X_\star\}
\end{equation}
be the set of independent random processes generated at the respective source nodes in $G$.

Let $\text{In}_{\widetilde{G}}(i)$ denote the edges $e \in \widetilde{E}$ that are incident to the vertex $i \in \widetilde{V}$ in the simple graph $\widetilde{G}$ and
\begin{equation}
\label{eq:defn_a_i}
A_i = \{X_i\} \cup \{X_e : e \in \text{In}_{\widetilde{G}}(i)\}.
\end{equation}
We use $A_i^c$ to denote $X \setminus A_i$.% = A_i^c$ for brevity.

Our sum-network $G$ can be constructed as follows. We first include the source node set $S$, the terminal node set $T$ and the bottleneck edges $B=\{e_1,e_2, \hdots , e_b\}$ where $b = |\widetilde{V}|$ and each edge is of capacity $\alpha$. Next, we follow the construction algorithm below  where edges (each of capacity $\alpha$) are included.\\
\underline{\it Construction Algorithm 1}\\
\begin{enumerate} \itemsep 1em
%\item $B=\{e_1,e_2, \hdots , e_b\}$ where $b = |\widetilde{V}|$
\item   Edges from sources to bottlenecks.
        \begin{enumerate}
			\item $(s_i,\text{tail}(e_j))$ if $X_i \in A_j ~ \text{for all} ~i,j \in \{1,2,\hdots,b\}$,
			\item $(s_e,\text{tail}(e_j))$ if $X_e \in A_j ~ \text{for all} ~j \in \{1,2,\hdots,b\}$.% and
            %\item $s_\star, \text{tail}(e_j)$ for all $j \in \{1,2,\hdots,b\}$.
		\end{enumerate}

\item Edges from bottlenecks to terminals.
        \begin{enumerate}
		\item $(\text{head}(e_i),t_i)  ~ \text{for all}~ i=\{1,2,\hdots,b\}$,
		\item $(\text{head}(e_i),t_e)$ and $(\text{head}(e_j),t_e) ~ \text{for all}~ e=(i,j) \in \widetilde{E}$, and
        \item $(\text{head}(e_i),t_{\star}) ~ \text{for all~} i \in \{1,2,\hdots,b\}$.
        \end{enumerate}

\item Direct edges.
        \begin{enumerate}
        \item $(s_i,t_j)$ if $X_i \in A_j^c ~ \text{for all} ~i,j \in \{1,2,\hdots,b\}$ and $(s_e, t_j)$ if $X_e \in A_j^c$ for $j \in \{1,2,\hdots,b\}$.
        \item For all $e=(i,j) \in \widetilde{E}$, $(s_k,t_e)$ if $X_k \in A_i^c \cap A_j^c$ and $(s_{e'}, t_e)$ if $X_{e'} \in A_i^c \cap A_j^c$.
        %\item $(s_e,t_e) ~ \text{for all} e \in \widetilde{E}$.
        \end{enumerate}

\end{enumerate}
Next, we illustrate an example of the above construction.
\begin{example}\label{eg:K4e}
 Let $\widetilde{G}$ be as shown in Figure \ref{fig:K4e}.
%\subsection{Example: $\widetilde{G}=K_3$, the complete graph on three nodes}
\begin{figure}[t]
\centering
\includegraphics[scale = 0.5]{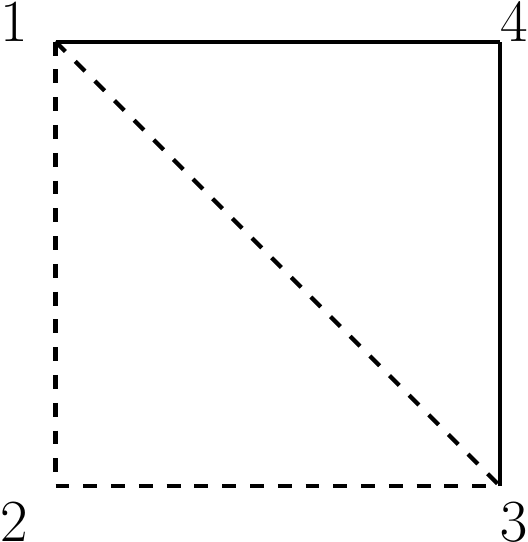}
\caption{ $\widetilde{G}$ for Example \ref{eg:K4e}. The dotted lines show the edge set $E_{\widetilde{Cyc}}$ for Construction 2.}
\label{fig:K4e}
\end{figure}
Then,
\begin{align*}
A_1 &= \{X_1,X_{(1,2)},X_{(1,3)},X_{(1,4)}\} \\
A_2 &= \{X_2,X_{(1,2)},X_{(2,3)}\} \\
A_3 &= \{X_3,X_{(1,3)},X_{(2,3)},X_{(3,4)}\} \\
A_4 &= \{X_4,X_{(1,4)},X_{(3,4)}\}
\end{align*}
The corresponding sum-network $G$ is shown in Figure \ref{fig:sum_network_K3}.
%The bottleneck edges in DAG $G$ are $B=\{e_1,e_2,e_3\}$. According to the construction we have, e.g.,
%\begin{align*}
%\text{In}_{G}(e_1)&=\{(w,e_1): w \in \{s_2,s_3,s_{(2,3)},s^\star\}\}, \\
%\text{In}_{G}(t_1)&=\{(w,t_1): w \in \{s_1,s_{(1,2)},s_{(1,3)},\text{head}(e_1)\}\},  \\
%\text{In}_{G}(t_{(1,2)})&=\{(w,t_{(1,2)}): w \in \{\text{head}(e_1),\text{head}(e_2),s_{(1,2)}\}\}
%\end{align*}
%
%%\begin{align*}
%%\text{In}_{\widetilde{G}}(e_1)&=\{(w,e_1): w \in \{s_2,s_3,s_{(2,3)},s^\star\}\} \\
%%\text{In}_{\widetilde{G}}(t_1)&=\{(w,t_1): w \in \{s_1,s_{(1,2)},s_{(1,3)},\text{head}(e_1)\}\} \\
%%\text{In}_{\widetilde{G}}(t_{(1,2)})&=\{(w,t_{(1,2)}): w \in \{\text{head}(e_1),\text{head}(e_2),s_{(1,2)}\}\}
%%\end{align*}
%and similarly for other edges. The constructed directed acyclic graph $G$ is shown in Figure \ref{fig:sum_network_K3}.
\begin{figure}
\centering
\includegraphics[scale=0.53]{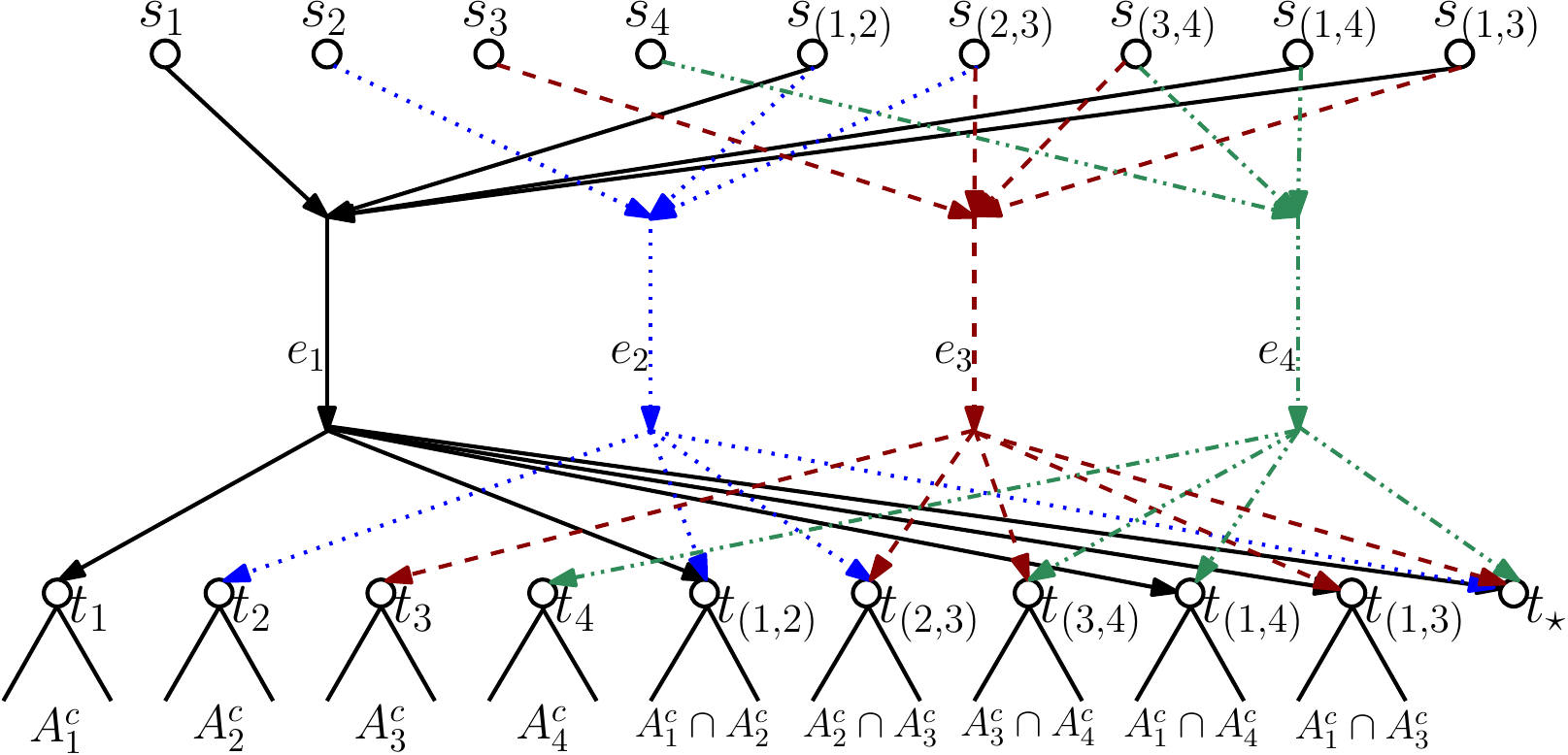}
\caption{Sum-network $G$ constructed from $\widetilde{G}$ in Fig. \ref{fig:K4e} via Construction 1. Edges $e_i, i = 1, \dots, 4$ represent the bottlenecks. The direct edges are specified by means of the set that appears below each terminal.}
\label{fig:sum_network_K3}
\end{figure}
\begin{figure}
\centering
\includegraphics[scale=0.53]{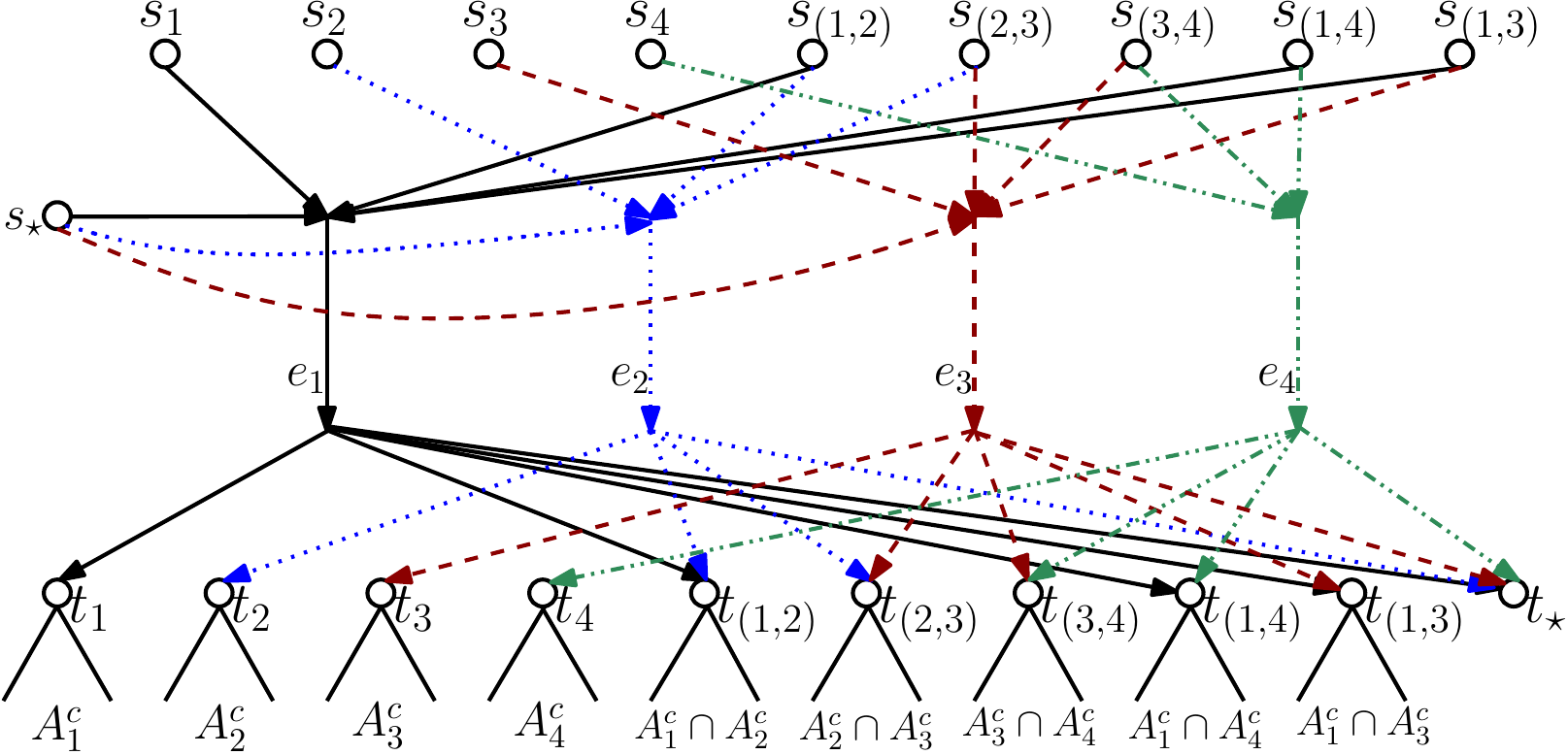}
\caption{Sum-network $G$ constructed from $\widetilde{G}$ in Fig. \ref{fig:K4e} via Construction 2 (with source $s_\star$). Edges $e_i, i = 1, \dots, 4$ represent the bottlenecks. The direct edges are specified by means of the set that appears below each terminal.}
\label{fig:sum_network_K3*}
\end{figure}
%Upon inspection, it can be observed that each source is connected to each terminal by at least one path of capacity $\alpha$.
\end{example}

\subsection{Construction 2}
In the second construction, we include the source $s_\star$, i.e., $S = S_1 \cup S_2 \cup \{s_\star\}$, so that we have source $X_\star$ in addition to the sources listed in eq. (\ref{eq:sources}). Recall that we assumed that $\widetilde{G}$ is not a tree. Let $\widetilde{Cyc} = (V_{\widetilde{Cyc}}, E_{\widetilde{Cyc}})$ be a subgraph of $\widetilde{G}$ corresponding to the {\it shortest cycle} in $\widetilde{G}$; $\widetilde{Cyc}$ may not be unique. The following is a useful fact (proof appears in the Appendix).
\begin{claim}
\label{claim:shortest_cycle}
Suppose that nodes $i, j \in V_{\widetilde{Cyc}}$. Then either $(i,j) \in E_{\widetilde{Cyc}}$ or $(i,j) \notin \widetilde{E}$.
\end{claim}

%\begin{equation*}
%X=\{X_1, X_2, \hdots , X_{|\widetilde{V}|}\} \cup \{X_e : e \in \widetilde{E}\} \cup \{X_\star\}
%\end{equation*}
%be the set of independent random processes generated at the respective source nodes in $G$.
%The
%Let $\text{In}_{\widetilde{G}}(i)$ denote the edges $e \in \widetilde{E}$ that are incident to the vertex $i \in \widetilde{V}$ in the simple graph $\widetilde{G}$.
The set of random processes $A_i$ is defined as follows.
\begin{equation*}
A_i = \begin{cases}
        \{X_i\} \cup \{X_e : e \in \text{In}_{\widetilde{G}}(i)\} \cup \{X_\star\} & \text{~if~} i \in V_{\widetilde{Cyc}},\\
        \{X_i\} \cup \{X_e : e \in \text{In}_{\widetilde{G}}(i)\}  & \text{~otherwise.}
\end{cases}
\end{equation*}
Following the definition of the sets $A_i$, the steps outlined in Construction 1 can be followed to construct most of the sum-network $G$. Following this, the additional source $s_\star$ is connected to each $e_i$ where $i \in V_{\widetilde{Cyc}}$. Each terminal in $T$ that does not have a path from $s_\star$ to it, is provided a direct edge $s_\star$ to it. This concludes the construction of $G$.
For an example, check Figure \ref{fig:sum_network_K3*} for a sum-network constructed from Figure \ref{fig:K4e} with the choice of $\widetilde{Cyc}$ as indicated in caption.
\subsection{Upper bound on the capacity of $G$}
We derive an upper bound on the capacity of $G$ assuming that we followed Construction 2 so that source $s_\star$ exists. The corresponding upper bound for Construction 1 follows in a similar manner.
Suppose that there exists a $(r,l)$ fractional network code assignment $\phi_e, e \in E$ and decoding functions $\psi_{t}, t \in T$ so that all the terminals in $T$ can recover the sum of sources denoted by $Z = \sum_{i \in \widetilde{V}}X_i + \sum_{e \in \widetilde{E}} X_e+ X_\star$. As such we have three types of sources, corresponding to (i) the vertices of $\widetilde{G}$, (ii) the edges of $\widetilde{G}$ and (iii) the starred source. We use a generic index to refer to all these types of sources. For instance, $\sum_{\beta \in A_k} X_\beta$ would equal $\sum_{i:X_i \in A_k, i \in \widetilde{V}}X_i + \sum_{e:X_e \in A_k, e \in \widetilde{E}} X_e+ \mathbbm{1}_{\{k \in V_{\widetilde{Cyc}}\}} X_\star$, where $\mathbbm{1}$ is the indicator function.
\begin{lemma}
\label{lemma:comp_1}
For terminal $t \in T$, if edge $(\text{head}(e_k),t)$ exists, then terminal $t$  can compute $\sum_{\beta \in A_k}X_\beta$ from the function value on $e_k$.%the information (function value) transmitted on the bottleneck edge $e_i$.
\end{lemma}

\begin{proof}
Consider terminal $t_k$ for $k \in \widetilde{V}$. It is connected to $\text{head}(e_k)$. By assumption, it is able to recover $Z$ from the information transmitted on its incoming edges. %Suppose that the function value on $e_i$ is not in one-to-one correspondence with the values  $\sum_{i:X_i \in A_i^c}X_i$.
In the discussion below we let $\phi_{e_k}(X)$ refer to the value that the network code assigns on edge $e_k$, where we emphasize that $\phi_{e_k}(X)$ only depends on sources in the set $A_k$.
%Suppose there is no valid transformation (i.e., a many-to-one mapping) from $\phi_{e_k}(X)$ to the sum $\sum_{\beta \in A_k}X_\beta$.
Suppose we cannot decode the sum $\sum_{\beta \in A_k}X_\beta$ from the value of $\phi_{e_k}(X)$.
This implies that we can find two different instantiations of source symbols $\mathbf{x}$ and $\mathbf{x}'$ such that
%If that is the case, then consider the particular two different sets of source symbols $X$ and $X'$ such that
\begin{itemize}
\item $\phi_{e_k}(\mathbf{x})$  = $\phi_{e_k}(\mathbf{x}')$ but $\sum_{\beta \in A_k} \mathbf{x}_\beta \neq \sum_{\beta \in A_k} \mathbf{x}_\beta'$, and
\item $\mathbf{x}_\beta=\mathbf{x}_\beta'$ for $\beta \in A_k^c$.
%\item $\sum_{\beta \in A_k} \mathbf{x}_\beta \neq \sum_{\beta \in A_k} \mathbf{x}_\beta'$
%\item $\phi_{e_i}(\mathbf{x})$  = $\phi_{e_i}(\mathbf{x}')$.
\end{itemize}
%\begin{itemize}
%\item $\mathbf{x}_i=\mathbf{x}_i' \text{~for~} i \text{~such that~} X_i \in A_k^c$, $\mathbf{x}_e=\mathbf{x}_e' \text{~for~} e \text{~such that~} X_e \in A_k^c$, and if $k \notin \widetilde{Cyc}$, $\mathbf{x}_\star = \mathbf{x}_\star'$.
%\item $\sum_{i:X_i \in A_k, i \in \widetilde{V}} \mathbf{x}_i + \sum_{e:X_e \in A_k, e \in \widetilde{E}} \mathbf{x}_e + \mathbb{1}_{k \in \widetilde{Cyc}} \mathbf{x}_\star \neq \sum_{i:X_i \in A_k, i \in \widetilde{V}} \mathbf{x}_i' + \sum_{e:X_e \in A_k, e \in \widetilde{E}} \mathbf{x}_e' + \mathbb{1}_{k \in \widetilde{Cyc}} \mathbf{x}_\star'$
%$\sum_{i:X_i \in A_k} X_i  \neq \sum_{i:X_i' \in A_i'^c}X_i'$
%\item $\phi_{e_i}(\mathbf{x})$  = $\phi_{e_i}(\mathbf{x}')$.
%\end{itemize}
Thus, even though $Z \neq Z'$, $\psi_{t_k}(\mathbf{x})$  = $\psi_{t_k}(\mathbf{x}')$ as $\phi_{e}(\mathbf{x}) = \phi_e(\mathbf{x}')$ for all $e  \in \text{In}(t_k)$. Thus, the terminal $t_k$ is unable to compute the sum which is a contradiction.
\end{proof}
In a similar manner the following lemma can be shown to hold.
\begin{lemma}
\label{lemma:comp_2}
Any terminal $t_e$ for $e=(i,j)\in \widetilde{E}$ is able to compute $\sum_{\beta \in A_i \cup A_j} X_\beta$ from the function values on edges $e_i$ and $e_j$ in $G$.% $\phi_{e_i}(X)$ and $\phi_{e_j}(X)$.
\end{lemma}
%\textit{Proof}: Consider terminal $t_e$. By assumption, it is able to compute $\sum_i X_i$ from the information on $\text{In}(t_{i,j})$. Now, $e_i$ and $e_j$ belong to $\text{In}(t_e)$.
%Then, similar to the previous lemma, there is a valid transformation from the function values on the pair of edges $(e_i,e_j)$ to the value of $\sum_{i:X_i \in (A_i \cap A_j)^c}X_i$. For if not, there is a particular pair of function values on the edges $(e_i,e_j)$ which corresponds to different values of $\sum_{i:X_i \in  (A_i \cap A_j)^c}X_i$.

%If that is the case, then consider the particular two different sets of source symbols $X$ and $X'$ such that:
%\begin{itemize}
%\item $X_i=X_i' ~\text{if}~ X_i \in A_i \cap A_j, X_j' \in A_i' \cap A_j'$
%\item $\sum_{i:X_i \in  (A_i \cap A_j)^c}X_i \neq \sum_{i:X_i' \in  (A_i' \cap A_j')^c}X_i'$
%\item $(\phi_{e_i},\phi_{e_j})$  = $(\phi_{e_i}',\phi_{e_j}')$
%\end{itemize}
%Then, even though $\sum X_i \neq \sum X_i'$, $\psi_{t_e}$  = $\psi_{t_e}'$ as $\phi_{e}~ \forall e  \in \text{In}(t_{i,j})$ is the same in both the cases. Hence, the terminal $t_e$ is unable to compute the sum and we have a contradiction.
%From lemma 1, we can compute $\sum_{i:s_i \in X_i}s_i$
%\textit{Proof}: The proof is similar to the proof of lemma 1.
\begin{theorem}
\label{thm:upper_bd}
The rate of a fractional network coding solution for the sum-network $G$ constructed by the procedure above is upper bounded by $\frac{\alpha |\widetilde{V}|}{|\widetilde{E}|+|\widetilde{V}|+1}$.
\end{theorem}
%\begin{equation*}
%\frac{\alpha |\widetilde{V}|}{|\widetilde{E}|+|\widetilde{V}|+1}
%\end{equation*}
\begin{proof}
We assume that there is a valid network code such that each terminal is able to compute the sum $Z$.  %Note that terminal $t_\star$ is connected to all the bottleneck edges.
%For the sake of simplicity, we assume that $|\widetilde{Cyc}| = c$ and the first $c$ nodes in $\widetilde{G}$ participate in $\widetilde{Cyc}$. Thus, sets $A_1, \dots, A_c$ contain $X_\star$ while $A_{c+1}, \dots, A_b$ do not.
Consider a terminal $t_e$ where $e = (i,j) \in \widetilde{E}$. In this case $t_e$ can compute (from Lemma \ref{lemma:comp_1} and \ref{lemma:comp_2}), the following partial sums
\begin{itemize}
\item $\sum_{\beta \in A_i} X_\beta$.
\item $\sum_{\beta \in A_j} X_\beta$.
\item $\sum_{\beta \in A_i \cup A_j} X_\beta$.
\end{itemize}
Thus, it can compute $\sum_{\beta \in A_i} X_\beta + \sum_{\beta \in A_j} X_\beta - \sum_{\beta \in A_i \cup A_j} X_\beta$. Now, if $e$ participates in the cycle $\widetilde{Cyc}$ we obtain $X_e + X_\star$, which can be observed by noting that $A_i \cap A_j = \{X_e,X_\star\}$. It can also be seen that if edge $(i,j) \notin E_{\widetilde{Cyc}}$ then at least one of $i$ or $j$ do not participate in $\widetilde{Cyc}$ (\emph{cf.} Claim \ref{claim:shortest_cycle}). In this case the previous operation would simply provide $X_e$. %\aditya{need a proof, uses the fact that we are picking shortest cycle}.
Terminal $t_\star$ is connected to all the bottleneck edges. Thus, it can obtain the following symbols.
\begin{align*}
Y_{1e} &= X_e + X_\star~  \text{if}~ e \in E_{\widetilde{Cyc}}\\
Y_{1e} &= X_e ~ \text{if}~ e \notin E_{\widetilde{Cyc}}
\end{align*}
Following this step, $t_\star$ can compute $\sum_{\beta \in A_k} X_\beta - \sum_{e:X_e \in A_k} Y_{1e}$. Note that $A_k$ contains only one source corresponding to the vertices, namely $X_k$. Now, if $k \notin V_{\widetilde{Cyc}}$, then it is evident that $Y_{1e} = X_e$ for all $e \in \text{In}_{\widetilde{G}}(k)$, so that $\sum_{\beta \in A_k} X_\beta - \sum_{e:X_e \in A_k, e \in \widetilde{E}} Y_{1e} = X_k$. Alternatively, if $k \in V_{\widetilde{Cyc}}$, then $\sum_{\beta \in A_k} X_\beta - \sum_{e:X_e \in A_k, e \in \widetilde{E}} Y_{1e} = X_k - X_\star$. Thus at the end of this operation, $t_\star$ has the following symbols.
\begin{align*}
Y_{2k} &= X_k - X_\star~ \text{if $k \in V_{\widetilde{Cyc}}$}\\
Y_{2k} &= X_k ~ \text{if $k \notin V_{\widetilde{Cyc}}$}
\end{align*}
Note that a cycle has the same number of edges and nodes. Thus, $\sum_{k \in \widetilde{V}} Y_{2k} + \sum_{e \in \widetilde{E}} Y_{1e} = \sum_{i \in \widetilde{V}}X_i + \sum_{e \in \widetilde{E}} X_e$. However, terminal $t_\star$ already knows $Z$, thus it can compute $X_\star$ and consequently the value of all the sources.
%From Lemma \ref{lemma:comp_1} and \ref{lemma:comp_2} we can conclude for instance that it can computerNow consider the terminal $t_\star$, which is assumed to be able to compute the desired sum $Z = \sum_{i=1}^{|\widetilde{V}|} X_i + \sum_{e \in \widetilde{E}} X_e + X_*$. %It is connected to all the bottleneck-edges through $\alpha$-capacity links and receives the information being transmitted on those edges.
%
%%By assumption terminal $t^\star$ can compute $\sum X_i$ and
%From Lemma \ref{lemma:comp_2} it can compute $\sum_{i:X_i \in (A_i \cap A_j)^c}X_i$. Hence, for every $X_e ~ e \in \widetilde{E}$, $t_\star$ can compute it by computing
%\begin{equation*}
% X_e = Z-\sum_{i:X_i \in (A_i \cap A_j)^c}X_i.
%\end{equation*}
%
%%By assumption terminal $t^\star$ can compute $\sum X_i$ and
%Likewise, from Lemma \ref{lemma:comp_1} it can compute $\sum_{i:X_i \in A_i^c}X_i$, and it already knows $X_e$ for $e \in \widetilde{E}$. Hence, for every $X_i ~ i \in \widetilde{V}$, $t_\star$ can compute its value by
%\begin{equation*}
%X_i = Z -\sum_{i:X_i \in A_i^c}X_i-\sum_{e: e \in \text{In}_{\widetilde{G}}(i) }X_e.
%\end{equation*}
%
%%If we had not included the source $s^*$ while constructing $\widetilde{G}$, then $t^\star$ is able decode all the source symbols at this point and hence is able to compute the sum.
%
%%On the other hand, if we had included the source $s^*$ in the construction,
%Finally, $t_\star$ can compute $X_\star$ by the operation
%\begin{equation*}
%X_* = Z - \sum_{e: e \in \widetilde{E}}X_e - \sum_{i: i \in \widetilde{V}}X_i.
%\end{equation*}
Therefore, under the assumption that all the terminals can compute the sum of sources, we can conclude that $t^\star$ is able to compute all the individual source symbols present in the sum-network. %If we had not included the extra source while constructing $\widetilde{G}$, then this step would be unnecessary.

It can be observed that the minimum cut between the set of all the sources and $t_\star$ is $\alpha |\widetilde{V}|$. Then under a valid fractional $(r,l)$-network code we must have
\begin{align}
\left(q^{l}\right)^{\alpha |\widetilde{V}|} &\geq \left( q^r \right) ^{|\widetilde{E}|+|\widetilde{V}|+1} \\
\implies \frac{r}{l} &\leq \frac{\alpha |\widetilde{V}|}{|\widetilde{E}|+|\widetilde{V}|+1}.
\label{UB1}
\end{align}
%which can in general be made $\leq 1$.
This concludes the proof.
\end{proof}
Thus for the sum network constructed in Example \ref{eg:K4e} we must have, for a $(r,l)$-network coding solution
\begin{equation}\label{eq:K3*_upper_bound}
\frac{r}{l} \leq \frac{4\alpha}{4+5+1} = \frac{2\alpha}{5}.
\end{equation}
\begin{remark}
It can be seen that following the line of proof above for the case of Construction 1 (without source $s_\star$), one arrives at a capacity upper bound of $\frac{\alpha |\widetilde{V}|}{|\widetilde{E}|+|\widetilde{V}|}$.
\end{remark}
\begin{remark}
In Construction 2, we chose to connect the starred source to only a carefully chosen subset of the bottleneck edges. If instead we had connected it to all the bottleneck edges (for instance), the starred terminal could only recover the source, under conditions on the characteristic of the field $\mathcal{A}$. This dependency on characteristic is evaluated and shown for an example sum-network in the Appendix. Our choice of the subset of bottleneck edges avoids this dependence on the field characteristic.

There is further work done in \cite{tripathyR15}, in which it is shown that the computation capacity (taking into account both linear and non-linear network codes) of a sum-network is strongly dependent on the finite field $\mathcal{A}$ chosen for computation and communication.
\end{remark}
%
%Similarly, in the case when the source $s_\star$ is not included, it can be observed that we have
%\begin{align}
%\frac{r}{l} &\leq \frac{\alpha b}{|\widetilde{E}|+|\widetilde{V}|}.
%\label{UB2}
%\end{align}
%\end{proof}
%\subsection{Example: $\widetilde{G}=K_4$, the complete graph on four nodes}
%\begin{center}
%\begin{tikzpicture}
% \def \radius {2cm}
% \def \margin {8}
% \def \n {4}
% \foreach \s in {1,...,\n}
%  \node[draw, circle] (\s) at ({360/\n * (\s - 1)}:\radius) {\s};
% \foreach \s in {1,...,\n}
%  \foreach \t in {\s,...,\n}
%   \draw (\t) -- (\s);
%\end{tikzpicture}
%\end{center}
%Then, \begin{align*}
%A_1 &= \{X_1,X_{(1,2)},X_{(1,3)},X_{(1,4)}\} \\
%A_2 &= \{X_2,X_{(1,2)},X_{(2,3)},X_{(2,4)}\} \\
%A_3 &= \{X_3,X_{(1,3)},X_{(2,3)},X_{(3,4)}\} \\
%A_4 &= \{X_4,X_{(1,4)},X_{(2,4)},X_{(3,4)}\}
%\end{align*}
%
%The bottleneck edges in DAG $G$ are $B=\{e_1,e_2,e_3,e_4\}$. Also,
%\begin{align*}
%\text{In}(e_1)&=\{(w,e_1): w \in \{s_2,s_3,s_4,s_{(2,3)},s_{(2,4)},s_{(3,4)},s^\star\}\} \\
%\text{In}(t_1)&=\{(w,t_1): w \in \{s_1,s_{(1,2)},s_{(1,3)},s_{(1,4)},\text{head}(e_1)\}\} \\
%\text{In}(t_{(1,2)})&=\{(w,t_{(1,2)}): w \in \{\text{head}(e_1),\text{head}(e_2),s_{(1,2)}\}\}
%\end{align*}
%and similarly for other edges.

\subsection{Achievability scheme for $G$}
For certain classes of undirected simple connected graphs, the corresponding sum-network is such that we can demonstrate a linear network coding scheme that achieves the upper bound in Theorem \ref{thm:upper_bd}. Towards this end, we first assign non-negative integers to variables $m_{(i,j)}(i)$ and $m_{(i,j)}(j)$, for all $(i,j) \in \widetilde{E}$, that satisfy certain constraints. The constraints depend on whether we use Construction 1 or 2 for constructing $G$ and are given below.\\
%We now construct a $(r,l)$ fractional linear network code for for a class of graphs, i.e. we demonstrate a linear coding scheme which achieves the upper bound in section~\ref{constr}.
%Consider the case when $\widetilde{G}$ is connected and $\alpha = 1$.
%We assign two variables $m_{(i,j)}(i)$ and $m_{(i,j)}(j)$ for each edge $(i,j) \in \widetilde{E}$ and thus have a total of $2|\widetilde{E}|$ variables. %Let $\mathbb{Z}_+$ be the set of positive integers. %We assume that $\widetilde{G}$ also satisfies the condition that there exists an assignment to these $2|\widetilde{E}|$ variables which satisfies the following constraints
%We assume that the constraints outlined in \ref{opt-P} have a feasible solution, i.e. there exists a solution to
%\begin{align}
%& m_{(i,j)}(i) \in \mathbb{Z_+} \cup \{0\}~~ \forall (i,j) \in \widetilde{E}, i \in \widetilde{V} \\
%& m_{(i,j)}(i)+m_{(i,j)}(j)=b ~~ \forall (i,j) \in \widetilde{E} \\
%& \sum_{j : (i,j) \in \text{In}_{\widetilde{G}}(i)}m_{(i,j)}(i) \leq |\widetilde{E}|+1 ~\forall i \in \widetilde{V}
%\label{eq: cond_feas}
%\end{align}
%\aditya{feasibility opt will need to change for Construction 2}
\underline{\it Feasible solution for Construction 1.}
\begin{align} \label{opt-1} \tag{CODE-FEAS-1}
% \text{find}~& m_{(i,j)}(i)~~ \forall (i,j) \in \widetilde{E}, i \in \widetilde{V}\\
%\nonumber \text{s.t. }& m_{(i,j)}(i) \in \mathbb{Z_+} \cup \{0\}\\
\nonumber & m_{(i,j)}(i)+m_{(i,j)}(j)=b, ~~ \forall (i,j) \in \widetilde{E}. \\
\nonumber & \sum_{j : (i,j) \in \text{In}_{\widetilde{G}}(i)}m_{(i,j)}(i) \leq |\widetilde{E}|, ~~\forall i \in \widetilde{V}.
\end{align}
\underline{\it Feasible solution for Construction 2.}
\begin{align} \label{opt-2} \tag{CODE-FEAS-2}
\nonumber & m_{(i,j)}(i)+m_{(i,j)}(j)=b, ~~ \forall (i,j) \in \widetilde{E}. \\
\nonumber & \sum_{j : (i,j) \in \text{In}_{\widetilde{G}}(i)}m_{(i,j)}(i) \leq |\widetilde{E}|+1, ~~\forall i \in \widetilde{V}.\\
\nonumber & \sum_{i \in V_{\widetilde{Cyc}}}\left[ |\widetilde{E}|+1 -  \sum_{j : (i,j) \in \text{In}_{\widetilde{G}}(i)}m_{(i,j)}(i) \right] \geq b.
\end{align}
As will be evident shortly, the existence of the variables $m_{(i,j)}(i)$ allow us to construct the achievability scheme that matches the upper bound in Theorem \ref{thm:upper_bd}. Such feasible assignments do not exist for all graphs.
However, for a large class of graphs, we can in fact arrive at an assignment. For instance, for a simple, regular graph where $|\widetilde{V}|$ is even, it can be seen that $m_{(i,j)}(i) = |\widetilde{V}|/2, ~ \forall (i,j) \in \widetilde{E}$
results in a feasible assignment for \ref{opt-1}. The following claim can be shown (the proof appears in the Appendix).
\begin{claim}
\label{claim:feas}
If $\widetilde{G}$ is a regular graph or a biregular bipartite graph, then \ref{opt-1} has a feasible solution. %\aditya{can we have a corresponding result for second one as well}
\end{claim}

\subsubsection{ Linear network code for Construction 1}
We construct a $(r,l)$ network code for the sum-network $G$ constructed above, by choosing $r=|\widetilde{V}|$ and $l=|\widetilde{V}|+|\widetilde{E}|$.
%, we must specify the edge and decoding functions as outlined in section~\ref{setup}.
In the discussion below, we assume that $\alpha = 1$. However, it will be evident that the case of higher $\alpha$ is a simple extension of the scheme for $\alpha = 1$.
Here we describe the encoding function $\phi_{e_i}$ for bottleneck edge $e_i$, and assign the other edges to simply forward the messages that they receive.%  (i.e. they just serve as repeaters).

We number the edges in $\widetilde{E}$ in an arbitrary order, so that the edges can be indexed as $\tilde{e}_1, \dots, \tilde{e}_{|\widetilde{E}|}$. Then the source processes can be represented as a vector $\mathbf{\bar{X}}$ of dimension $r(|\widetilde{V}| + |\widetilde{E}|) $ as shown below.
\begin{equation*}
\mathbf{\bar{X}}=[X_1^T ~ X_2^T ~\hdots~ X_{|\widetilde{V}|}^T ~X_{\tilde{e}_1}^T ~X_{\tilde{e}_2}^T ~\hdots ~X_{\tilde{e}_{|\widetilde{E}|}}^T ]^T
\end{equation*} where $X_i, i = 1, \dots, |\widetilde{V}|$, $X_{\tilde{e}_i}, i = 1, \dots, |\widetilde{E}|$  are each of dimension $r \times 1$.
%Each element of $\mathbf{\bar{X}}$ is  in $\mathcal{A}^r$ i.e., a vector of $r$ components.
Recall that $\text{tail}(e_i)$ is connected to the set of sources in $A_i$, thus $\phi_{e_i}(\mathbf{\bar{X}})$ is only a function of the sources in $A_i$. Moreover, $\text{head}(e_i)$ is connected to terminals $t_i$ and $t_e$, where $e \in \text{In}_{\widetilde{G}}(i)$ and also $t_\star$.
%For ease of presentation, we order the edges in $In(i)$ in the following manner. Consider $\hat{e}_1 = (u, i)$ and $\hat{e}_2 = (v,i)$ (as $\widetilde{G}$ is undirected, no distinction is made between $(u,i)$ and $(i,u)$). Then $\hat{e}_1 \prec hat{e}_2$ if $u < v$.
We partition the vector $\phi_{e_i}(\mathbf{\bar{X}})$ as follows.
\begin{align*}
 \phi_{e_i}(\mathbf{\bar{X}})^T = [ \phi_{e_i}(\mathbf{\bar{X}})^T_1 ~ \phi_{e_i}(\mathbf{\bar{X}})^T_{\hat{e}_1} ~\hdots~ \phi_{e_i}(\mathbf{\bar{X}})^T_{\hat{e}_{|\text{In}(i)|}} ]
\end{align*}
where $\phi_{e_i}(\mathbf{\bar{X}})_1$ is of dimension $r \times 1$. Furthermore, $\hat{e}_1, \dots, \hat{e}_{|\text{In}_{\widetilde{G}}(i)|} \in \text{In}_{\widetilde{G}}(i)$
%such that $\hat{e}_1 \prec \hat{e}_2 \prec \hdots \prec \hat{e}_{|In(i)|}$
and $\phi_{e_i}(\mathbf{\bar{X}})_{\hat{e}_j}$ is of dimension $m_{\hat{e}_j}(i) \times 1$ ($m_{\hat{e}_j}(i)$ is obtained from \ref{opt-1}). %$m_{\hat{e}_j}(i)$ denotes the number of components
%that will be transmitted along edge $e_i$ of a linear combination of sources belonging to the set $A_i \cap A_u$
%of $X_{(u,i)}$ that will be transmitted along edge $e_i$ where $\hat{e}_j = (u,i) \in \text{In}_{\widetilde{G}}(i)$.
We set
\begin{align*}
\phi_{e_i}(\mathbf{\bar{X}})_1 = \sum_{\beta \in A_i} X_\beta
\end{align*}
%& \left[\begin{array}{cccccccc}
%B^i_1 & B^i_2 & \hdots & B^i_{|\widetilde{V}|} & B^i_{\tilde{e}_1} & \hdots & B^i_{\tilde{e}_{|\widetilde{E}|}} \\
%\end{array}\right] \mathbf{\bar{X}},
%\end{align*}
%where
%\begin{align*}
%B^i_j=
%\begin{cases}
%I_{r\times r} & \text{if}~ X_j \in A_i\\
%\mathbf{0}_{r \times r} & \text{otherwise}
%\end{cases}
%, \text{~for}~ 1 \leq j \leq |\widetilde{V}|, \text{~and}
%\end{align*}
%
%\begin{align*}
%B^i_{\tilde{e}_j}=
%\begin{cases}
%I_{r\times r} & \text{if}~ X_{\tilde{e}_j} \in A_i\\
%\mathbf{0}_{r \times r} & \text{otherwise}
%\end{cases}
%, \text{~for}~ \tilde{e}_j \in \widetilde{E}.

Next, for $\hat{e}_j = (u,i) \in \text{In}_{\widetilde{G}}(i)$, we set
\begin{align*}
\phi_{e_i}(\mathbf{\bar{X}})_{\hat{e}_j} =
\begin{cases}
[I_{m_{\hat{e}_j}(i) \times m_{\hat{e}_j}(i)} ~~ \mathbf{0}] X_{(u,i)} & \text{~if $i < u$} \\
[\mathbf{0} ~~I_{m_{\hat{e}_j}(i) \times m_{\hat{e}_j}(i)}] X_{(u,i)} & \text{~otherwise.}
\end{cases}
\end{align*}
%
%\begin{align*}
%& \phi_{e_i}(\mathbf{\bar{X}})_{\hat{e}_j} = \\
%& P X_{(u,i)}
%\end{align*}
%We define a matrix $P$ of dimension $m_{\hat{e}_j}(i) \times r$ as follows.
%\begin{align*}
% P = \begin{cases}
% [I_{m_{\hat{e}_j}(i) \times m_{\hat{e}_j}(i)} ~~ \mathbf{0}] & \text{~if $i < u$}\\
% [\mathbf{0} ~~I_{m_{\hat{e}_j}(i) \times m_{\hat{e}_j}(i)}] & \text{~otherwise.}
% \end{cases}
%\end{align*}

It is evident that the assignment for the encoding function $\phi_{e_i}(\mathbf{\bar{X}})$ discussed above is feasible since
\begin{enumerate}
\item it is a function only of sources in $A_i$, and
\item the row dimension of the transformation is \newline $r + \sum_{\hat{e}_j \in \text{In}(i)} m_{\hat{e}_j}(i) \leq |\widetilde{V}| + |\widetilde{E}|$.
\end{enumerate}
\begin{theorem}\label{thm:achiev}
%\textbf{Lemma 4}:
The sum of sources Z, %denoted $Z = \sum_{i = 1}^{|\widetilde{V}|} X_i + \sum_{e \in E} X_e + X^*$
can be computed at all the terminals with the encoding functions specified above.
\end{theorem}
\begin{proof}
%\textit{Proof}:
The terminal $t_i, i = 1, \dots, |\widetilde{V}|$ is connected to $\text{head}(e_i)$ in $G$ for $i \in \{1, \dots, |\widetilde{V}|\}$. Therefore, from $\phi_{e_i}(\mathbf{\bar{X}})_1$ it recovers $\sum_{\beta \in A_i} X_\beta$. %$\sum_{\{j: X_j \in A_i\}} X_j + \sum_{\{e: X_e \in A_i\}} X_e$.
Furthermore, it has access to the set of sources in $A_i^c$ via direct edges. Thus, it can compute the sum.

Next, consider a terminal $t_e$ where $e = (i,j) \in \widetilde{E}$ such that $i < j$. From the assignment above it can be seen that $t_e$ can recover $X_e$ since
%\begin{align*}
%X_{e}=
%\begin{bmatrix}
%\phi_{e_i}(\mathbf{\bar{X}})_e\\
%\phi_{e_j}(\mathbf{\bar{X}})_e
%\end{bmatrix}.
%\end{align*}
\begin{align*}
\begin{bmatrix}
\phi_{e_i}(\mathbf{\bar{X}})_e\\
\phi_{e_j}(\mathbf{\bar{X}})_e
\end{bmatrix}
= \begin{bmatrix}
[I_{m_{e}(i) \times m_{e}(i)} ~~ \mathbf{0}] X_{(i,j)} \\
[\mathbf{0} ~~ I_{m_{e}(j) \times m_{e}(j)}] X_{(i,j)}
\end{bmatrix}
= X_e,
\end{align*}
since $m_{e}(i) + m_{e}(j) = b = r$.
Next, note that $\sum_{\beta \in A_i} X_{\beta} + \sum_{\beta' \in A_j} X_{\beta'} - X_e = \sum_{\beta \in A_i \cup A_j} X_\beta$, since $A_i \cap A_j = X_e$. Moreover, $t_e$, gets all the sources $(A_i \cup A_j)^c$ via direct edges. Thus, it can compute the sum.

Finally, for $t_\star$, note that by previous arguments, it can recover any $X_{e}, ~ e \in |\widetilde{E}|$ from the bottleneck edges. Following this, it can recover
%\begin{align*}
$X_i=\phi_{e_i}(\mathbf{\bar{X}})_1-\sum_{e \in \text{In}_{\widetilde{G}}(i)} X_e$ for all $i \in \widetilde{V}$ and consequently the sum.
\end{proof}
%\aditya{this is not very convincing. Need to rephrase} This achievability scheme can be generalized to arbitrary integer edge capacity $\alpha$ by using the same linear code on $\alpha$ chunks of every source symbol separately, where each chunk $\in \mathcal{A}^r$.
\begin{example}
\label{eg:achiev_for _K_3}
We describe the edge functions and decoding procedure for the case when $\widetilde{G}=K_3$, i.e. the complete graph on three vertices. We note that the following assignment satisfies the constraints in \ref{opt-1}.
\begin{align*}
m_{(1,2)}(1)=1 ~~& m_{(1,2)}(2)=2, \\
m_{(1,3)}(1)=2 ~~& m_{(1,3)}(3)=1, \text{~and} \\
m_{(2,3)}(2)=1 ~~& m_{(2,3)}(3)=2.
\end{align*}
The linear encoding functions are shown in Fig. \ref{fig:achiev_scheme}.%eq. (\eqref{eqex}).
\end{example}

\begin{figure}[!t]
%\normalsize
%\setcounter{equation}{4}
\begin{align*}%\label{eqex}
\mathbf{\bar{X}}=
\begin{bmatrix}
X_1^T & X_2^T & X_3^T & X_{(1,2)}^T & X_{(1,3)}^T  & X_{(2,3)}^T
\end{bmatrix}^T
\end{align*}
\begin{align*}
\phi_{e_1}(\mathbf{\bar{X}})&=
\left[ \begin{array}{c c c c c c   }
I_3 & \mathbf{0}_{3 \times 3} & \mathbf{0}_{3 \times 3} & I_3 & I_3 & \mathbf{0}_{3 \times 3}\\
\mathbf{0} & \mathbf{0} & \mathbf{0} & [1 ~\mathbf{0}] & \mathbf{0} & \mathbf{0}\\
\mathbf{0} & \mathbf{0} & \mathbf{0} & \mathbf{0} & [I_2 ~\mathbf{0}] & \mathbf{0}
\end{array}\right ]\mathbf{\bar{X}}
\end{align*}
\begin{align*}
\phi_{e_2}(\mathbf{\bar{X}})&=
\left[ \begin{array}{c c c c c c  }
\mathbf{0}_{3 \times 3} & I_3 & \mathbf{0}_{3 \times 3} & I_3 & \mathbf{0}_{3 \times 3} & I_3\\
 \mathbf{0} & \mathbf{0} & \mathbf{0} & [\mathbf{0}~I_2] & \mathbf{0} & \mathbf{0}\\
\mathbf{0} & \mathbf{0}  & \mathbf{0}  & \mathbf{0} & \mathbf{0} & [1 ~\mathbf{0}]
\end{array}\right ] \mathbf{\bar{X}}
\end{align*}
\begin{align*}
\phi_{e_3}(\mathbf{\bar{X}})&=
\left[ \begin{array}{c c c c c c  }
\mathbf{0}_{3 \times 3} & \mathbf{0}_{3 \times 3} & I_3 & \mathbf{0}_{3 \times 3} & I_3 & I_3 \\
\mathbf{0} & \mathbf{0} & \mathbf{0} & \mathbf{0} & [\mathbf{0} ~1]   & \mathbf{0}\\
\mathbf{0}  & \mathbf{0} & \mathbf{0} & \mathbf{0} & \mathbf{0} & [\mathbf{0} ~I_2]
\end{array}\right ] \mathbf{\bar{X}}
\end{align*}
\vspace*{4pt}
\caption{\label{fig:achiev_scheme} Encoding functions for the sum-network constructed from $K_3$ (see discussion in Example \ref{eg:achiev_for _K_3}). It can observed that $t_{1,2}$ can be satisfied since can recover the first component of $X_{(1,2)}$ from $\phi_{e_1}(\mathbf{\bar{X}})$ and the remaining components from $\phi_{e_2}(\mathbf{\bar{X}})$, following which it can recover $X_1 + X_2 + X_{(1,2)} + X_{(2,3)} + X_{(1,3)}$. It has $X_3$ available via a direct edge. In a similar fashion, it can be verified that all terminals can be satisfied.}
\end{figure}

\subsubsection{Linear network code for Construction 2}
For the second construction, we have one additional source $X_\star$ and we construct a $(r,l)$ network code for the sum-network with $r=|\widetilde{V}|$ and $l=|\widetilde{V}|+|\widetilde{E}|+1$. Similar to the previous case we index the edges so that the source processes can be stacked in a vector $\mathbf{\bar{X}}$ of dimension $r(|\widetilde{V}|+|\widetilde{E}|+1)$ as follows
\begin{equation*}
\mathbf{\bar{X}}=[X_1^T ~ X_2^T ~\hdots~ X_{|\widetilde{V}|}^T ~X_{\tilde{e}_1}^T ~X_{\tilde{e}_2}^T ~\hdots ~X_{\tilde{e}_{|\widetilde{E}|}}^T ~X_\star^T]^T
\end{equation*} where $X_i, i = 1, \dots, |\widetilde{V}|$, $X_{\tilde{e}_i}, i = 1, \dots, |\widetilde{E}|$ and $X_\star$  are each of dimension $r \times 1$.
The constraints satisfied by variables $m_{(i,j)}(i)$ is as outlined in \ref{opt-2}. Now, for $i \notin V_{\widetilde{Cyc}}$, the construction of $\phi_{e_i}(\mathbf{\bar{X}})$ is exactly the same as in the previous one.

For $i \in V_{\widetilde{Cyc}}$, the vector $\phi_{e_i}(\mathbf{\bar{X}})$ also contains information from source $s_\star$ as follows
\begin{align*}
 \phi_{e_i}(\mathbf{\bar{X}})^T = [ \phi_{e_i}(\mathbf{\bar{X}})^T_1 ~ \phi_{e_i}(\mathbf{\bar{X}})^T_{\hat{e}_1} ~\hdots~ \phi_{e_i}(\mathbf{\bar{X}})^T_{\hat{e}_{|\text{In}(i)|}} ~\phi_{e_i}(\mathbf{\bar{X}})_\star^T ]
\end{align*}
where $\phi_{e_i}(\mathbf{\bar{X}})_1$ is of dimension $r \times 1$. Also, $\hat{e}_1, \dots, \hat{e}_{|\text{In}_{\widetilde{G}}(i)|} \in \text{In}_{\widetilde{G}}(i)$
and $\phi_{e_i}(\mathbf{\bar{X}})_{\hat{e}_j}$ is of dimension $m_{\hat{e}_j}(i) \times 1$. $\phi_{e_i}(\mathbf{\bar{X}})_\star^T$ is of dimension $w_i \times 1$ where
\begin{align*}
w_i=|\widetilde{E}|+1-\sum_{j : (i,j) \in \text{In}_{\widetilde{G}}(i)}m_{(i,j)}(i).
\end{align*}
Note that $w_i \geq 0$ for $i \in V_{\widetilde{Cyc}}$ owing to the second constraint of \ref{opt-2}.

$\phi_{e_i}(\mathbf{\bar{X}})_1$ and $\phi_{e_i}(\mathbf{\bar{X}})_{\hat{e}_j} ~\text{for}~ \hat{e}_j \notin E_{\widetilde{Cyc}}$ are defined exactly as in previous construction. For $\hat{e}_j = (u,i) \in E_{\widetilde{Cyc}}$ there is a slight modification as follows.
\begin{align*}
\phi_{e_i}(\mathbf{\bar{X}})_{\hat{e}_j} =
\begin{cases}
[I_{m_{\hat{e}_j}(i) \times m_{\hat{e}_j}(i)} ~~ \mathbf{0}] (X_{(u,i)}+X_\star) & \text{~if $i < u$,} \\
[\mathbf{0} ~~I_{m_{\hat{e}_j}(i) \times m_{\hat{e}_j}(i)}] (X_{(u,i)}+X_\star) & \text{~otherwise.}
\end{cases}
\end{align*}
%Also, $\phi_{e_i}(\mathbf{\bar{X}})_\star^T$ is defined as
%\begin{align*}
%& \phi_{e_i}(\mathbf{\bar{X}})_\star^T= \\
%& \begin{bmatrix}
%\mathbf{0}_{w_i \times r(|\widetilde{E}|+|\widetilde{V}|)} & \mathbf{0}_{w_i \times \sum_{j\in V_{\widetilde{Cyc}}}^{j<i}w_j} & I_{w_i} & \mathbf{0}_{w_i \times \sum_{j\in V_{\widetilde{Cyc}}}^{j \leq i}w_j}
%\end{bmatrix}
%\mathbf{\bar{X}}
%\end{align*}

We let $\gamma_{w_i}= \sum_{\{j: j\in V_{\widetilde{Cyc}}, j < i\}} w_j$ and $\gamma_{w_i}'= \sum_{\{j: j\in V_{\widetilde{Cyc}}, j > i\}} w_j$. Then, $\phi_{e_i}(\mathbf{\bar{X}})_\star$ is defined as
\begin{align}
\label{eq:starred_source_recovery}
& \phi_{e_i}(\mathbf{\bar{X}})_\star =
\begin{bmatrix}
\mathbf{0}_{w_i \times r(|\widetilde{E}|+|\widetilde{V}|)} & \mathbf{0}_{w_i \times \gamma_{w_i}} & I_{w_i} & \mathbf{0}_{w_i \times \gamma_{w_i}'}
\end{bmatrix}
\mathbf{\bar{X}}
\end{align}

Such an assignment for the encoding function $\phi_{e_i}(\mathbf{\bar{X}})$ for $i \in V_{\widetilde{Cyc}}$ is feasible as
\begin{enumerate}
\item $\text{tail}(e_i)$ is connected to the starred source $s_\star$, and
\item the row dimension of the transformation is
\begin{align*}
& r + \sum_{\hat{e}_j \in \text{In}_{\widetilde{G}}(i)} m_{\hat{e}_j}(i)+ w_i \\ & =r + \sum_{\hat{e}_j \in \text{In}(i)} m_{\hat{e}_j}(i)+(|\widetilde{E}|+1-\sum_{j : (i,j) \in \text{In}_{\widetilde{G}}(i)}m_{(i,j)}(i))  \\ & =|\widetilde{E}|+|\widetilde{V}|+1.
\end{align*}
\end{enumerate}

\begin{theorem}\label{thm:achiev_starred}
The sum of sources Z, can be computed at all the terminals with the encoding functions specified above.
\end{theorem}
\begin{proof}
For terminals other than $t_\star$ and $t_e$ where $e \in E_{\widetilde{Cyc}}$, the sum can be computed in the same fashion as described in Theorem \ref{thm:achiev}.

Consider terminal $t_e$ where $e=(i,j) \in E_{\widetilde{Cyc}}$ such that $i<j$. $t_e$ is connected to $\text{head}(e_i)$ and $\text{head}(e_j)$. $t_e$ can recover the sum $X_e+X_\star$ by
\begin{align*}
X_e+X_\star=
\begin{bmatrix}
\phi_{e_i}(\mathbf{\bar{X}})_e\\
\phi_{e_j}(\mathbf{\bar{X}})_e
\end{bmatrix}.
\end{align*}
Note that $\sum_{\beta \in A_i} X_{\beta} + \sum_{\beta' \in A_j} X_{\beta'} - (X_e+X_\star) = \sum_{\beta \in A_i \cup A_j} X_\beta$, since $A_i \cap A_j = \{X_e,X_\star \}$. Moreover, $t_e$, gets all the sources $(A_i \cup A_j)^c$ via direct edges. Thus, it can compute the sum.

Terminal $t_\star$ can obtain $X_\star$ as follows.
\begin{align*}
X_\star^T=
\begin{bmatrix}
\phi_{e_{i_1}}(\mathbf{\bar{X}})_\star^T &
\phi_{e_{i_2}}(\mathbf{\bar{X}})_\star^T &
\hdots &
\phi_{e_{i_{|V_{\widetilde{Cyc}}|}}}(\mathbf{\bar{X}})_\star^T
\end{bmatrix},
\end{align*}
where $i_1 < i_2 < \hdots < i_{|V_{\widetilde{cyc}}|}$ all belong to $V_{\widetilde{Cyc}}$. This is because $\sum_{i \in V_{\widetilde{Cyc}}} w_i = \sum_{i \in V_{\widetilde{Cyc}}} \left[|\widetilde{E}|+1-\sum_{j : (i,j) \in \text{In}_{\widetilde{G}}(i)}m_{(i,j)}(i) \right] \geq b$ owing to the constraints in \ref{opt-2}. Moreover, eq. (\ref{eq:starred_source_recovery}) shows that the different bottleneck edges recover all the disjoint subsets of $X_\star$. % = |V_{\widetilde{Cyc}}|(|\widetilde{E}|+1) - \sum_{i \in V_{\widetilde{Cyc}}} \sum_{j : (i,j) \in \text{In}_{\widetilde{G}}(i)}m_{(i,j)}(i) \geq b$ owing to the constraints in \ref{opt-2}.

For $e \notin E_{\widetilde{Cyc}}$, $t_\star$ can recover $X_e$ in the same way as $t_e$. For $e=(i,j) \in E_{\widetilde{Cyc}}~i<j$, $t_\star$ can recover $X_e $ as
\begin{align*}
\begin{bmatrix}
\phi_{e_i}(\mathbf{\bar{X}})_e\\
\phi_{e_j}(\mathbf{\bar{X}})_e
\end{bmatrix}
-X_\star
=X_e+X_\star-X_\star
\end{align*}
Thus, $t_\star$ can recover any $X_e, e\in \widetilde{E}$. Following this it can be seen that it can also recover $X_i$ for all $i \in \widetilde{V}$ and then the sum.%It can thus recover
%$X_i=\phi_{e_i}(\mathbf{\bar{X}})_1-\sum_{e \in \text{In}_{\widetilde{G}}(i)} X_e$ for all $i \in \widetilde{V}$. Knowing these and the value of $X_\star$, it can compute the sum.
\end{proof}

The linear network code described here assumed edges of unit capacity. The same scheme can be used in sum-networks where all the edges have integer edge capacity $\alpha > 1$. In order to see this, we think of every $\alpha$ capacity edge as a union of $\alpha$ unit capacity edges between the same two vertices. Then we can think of this modified network as a union of $\alpha$ sub-networks, each of which is topologically equivalent to the original network but consists of only unit capacity edges. We can use the linear network code described above to achieve the coding capacity $\frac{r}{l}$ on each of these unit edge capacity networks. Thus, at the very least, we will be able to transmit a sum of sources $\in \mathcal{A}^{\alpha r}$ to all the terminals by a repeated application of the same network code on all $\alpha$ sub-networks. Thus, this will be a rate $\frac{\alpha r}{l}$ solution to the original sum-network. However, as evinced by Theorem \ref{thm:upper_bd} this is also the upper bound on the rate and hence this repeated application of our linear network code achieves capacity.

\section{Comparison with existing results}
\label{sec:comp}
The work most closely related to ours is by Rai \& Das \cite{raiD13}. They showed that there exists a sum-network that has coding capacity $\frac{p}{q}$ for any $p,q$ and constructed a linear code that achieves capacity on an instance of such a sum-network. They also pointed out that finding sum-networks that have the same coding capacity but with fewer sources and terminals was an open problem. We now demonstrate that our approach is a strict generalization of the work of \cite{raiD13}.

Their construction for a network with capacity $p/q$ starts by constructing a base network that has a capacity of $1/q$. The base network only has unit capacity edges. The edges are replicated $p$ times to obtain a network with capacity $p/q$. We now show that our approach gives their result as a special case when $\widetilde{G}$ is chosen to be the complete graph on $2q-1$ vertices.
\begin{example}
Consider $\widetilde{G}$ to be a complete graph on $2q-1$ vertices. This graph has a feasible assignment to constraints in \ref{opt-1} by Claim \ref{claim:feas}. Construct the sum-network $G$ from $\widetilde{G}$ without using the extra source $s_\star$. Then, by using Theorems \ref{thm:upper_bd} and \ref{thm:achiev} the coding capacity of $G$ is $\frac{2q-1}{2q-1+\binom{2q-1}{2}} = \frac{1}{q}$. By replicating each edge $p$ times, we obtain a coding rate of $p/q$.
\end{example}
% To see this, we construct a sum-network with coding capacity $\frac{1}{q}$ using the same number of sources and terminals  as used by Rai-Das. For that, consider $\widetilde{G}$ to be a complete graph on $2q-1$ vertices. Construct the DAG sum-network $G$ from $\widetilde{G}$ without using the extra source $s^\star$. Then the upper bound on the coding rate is given by \eqref{UB2} which is
%\begin{align}
%\frac{r}{l} \leq \frac{2q-1}{2q-1+\binom{2q-1}{2}} = \frac{1}{q}
%\label{rai}
%\end{align}
%and the same can be achieved via the achievability scheme outlined. The number of sources and terminals used are $(2q-1)+\binom{2q-1}{2}$ and $(2q-1)+\binom{2q-1}{2}+1$ respectively, which is the same as in Rai-Das paper.
%One of the open questions posed in that paper was whether we can achieve the coding capacity rates using lesser number of sources and terminals.
However, it is important to note that our framework allows us more parameters, which can be chosen to obtain stronger results in the sense that specific coding rates can be achieved by using fewer sources and terminals.
%\begin{example}
%Suppose that the approach of \cite{raiD13} was used to construct a sum-network with capacity $2/5$. It can be verified that this requires a sum-network with $45$ sources and $46$ terminals. In contrast, in our approach we can choose $\widetilde{G}$ to be a 3-regular graph on 4 vertices. In this case, we obtain a sum-network $G$ (without source $s_\star$) with capacity $ \frac{4}{4 + 6} = \frac{2}{5}$. However, our sum-network only has $10$ sources and $11$ terminals
%\end{example}
\begin{example}
Suppose that the approach of \cite{raiD13} was used to construct a sum-network with capacity $2/5$. It can be verified that this requires a sum-network with $45$ sources and $46$ terminals. In contrast, in our approach we can choose $\widetilde{G}$ to be the graph in Fig. \ref{fig:K4e} (a feasible assignment for \ref{opt-2} can be easily derived). In this case, we obtain a sum-network $G$ (with source $s_\star$) with capacity $ \frac{4}{4 + 5 +1} = \frac{2}{5}$. However, our sum-network only has $9$ sources and $10$ terminals.
\end{example}

\begin{figure}
\centering
\includegraphics[scale=0.4]{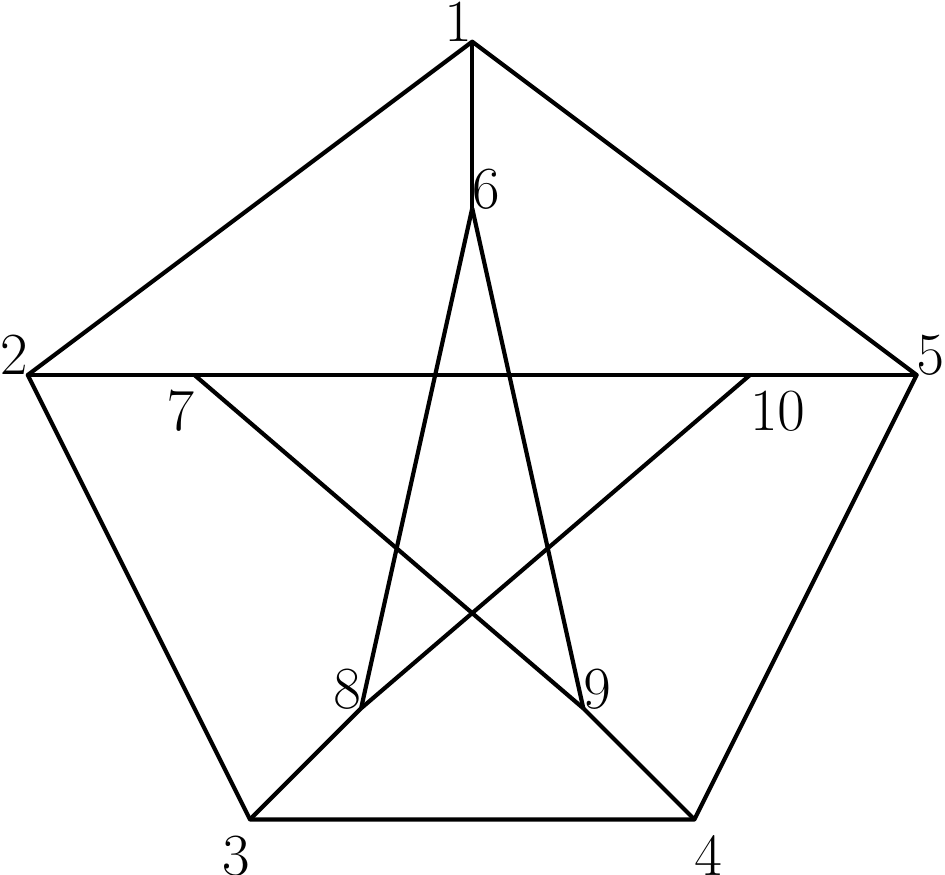}
\caption{Petersen graph.}
\label{fig:peter}
\end{figure}
As another example, we can construct a network with capacity $5/13$ that is significantly smaller.
\begin{example}
Let $\widetilde{G}$ be the Petersen graph \cite{diestel}, which is a 3-regular graph on 10 vertices (see Figure \ref{fig:peter}). Construct sum-network $G$ (with source $s_\star$). For this, choose the shortest cycle as $V_{\widetilde{Cyc}}=\{1,2,3,4,5\}$. The assignment for the variables in \ref{opt-2} can be done as follows
\begin{align*}
m_{(i,j)}(i) &=
\begin{cases}
4 & \text{~if $i \in V_{\widetilde{Cyc}}$ and $j \notin V_{\widetilde{Cyc}}$}, \\
6 & \text{~if $i \notin V_{\widetilde{Cyc}}$ and $j \in V_{\widetilde{Cyc}}$, and}  \\
5 & \text{~otherwise.}
\end{cases}\\
w_i &= 2 ~~\text{for all}~ i \in V_{\widetilde{Cyc}}.
\end{align*}
%We set
%\begin{align*}
%w_i=2 ~~\text{for all}~ i \in V_{\widetilde{Cyc}}.
%\end{align*}
The coding capacity of $G$ can be verified to be $\frac{10}{10+\frac{3 \times 10}{2}+1} = \frac{5}{13}$ and it has 26 sources and 26 terminals. However, the approach of \cite{raiD13} would need $2 \times 13 -1 + \binom{2 \times 13 -1}{2}=325$ sources and $326$ terminals.
\end{example}
%consider $\widetilde{G}$ to be the Petersen graph. It is a 3-regular graph on 10 vertices. Consider the construction of sum-network DAG $G$, for the case with an extra source symbol $s^\star$. Then the coding rate that can be achieved using this graph is given by \eqref{W2}
%\begin{equation}
%\frac{r}{l}=\frac{10}{10+\frac{3 \times 10}{2}+1}=\frac{10}{26}=\frac{5}{13}.
%\end{equation}
%We have achieved this coding rate using 26 sources. The complete graph approach would have required $2 \times 13 -1 + \binom{2 \times 13 -1}{2}=325$ sources. Thus, by using the extra parameters afforded by our framework, we are able to achieve a stronger result than the one described in Rai-Das paper.
To a certain extent we can analyze the parameters that can be obtained by considering regular graphs.
Let $\widetilde{G}$ be a $d$-regular graph on $b$ vertices, hence $d \leq b-1$.
%First we see that such a graph can always be constructed for all $d$ by uniformly arranging the $b$ vertices on a circle. If $d$ is even, then connect each vertex to $d/2$ neighbouring vertices on either side. If $d$ is odd, then connect each vertex to $\lfloor \frac{d}{2} \rfloor$ verices on either side and also to the vertex diametrically opposite.
%Then the two upper bounds on the coding rate that we can achieve through our construction are %(\eqref{UB1},\eqref{UB2})
Then the upper bounds that we have on the coding rate for Constructions 1 and 2 respectively are
\begin{equation}
\frac{r}{l} \leq \frac{\alpha b}{b+\frac{db}{2}} = \frac{2 \alpha}{2+d},
\label{W1}
\end{equation}
and
\begin{equation}
\frac{r}{l} \leq \frac{\alpha b}{b+\frac{db}{2}+1}.
\label{W2}
\end{equation}
Without loss of generality consider $p/q$ where $p$ and $q$ are coprime.
It can be seen that \eqref{W1} can achieve any such $p/q$ by setting $\alpha = p$ and $d= 2q-2$. The resulting $G$ is a $2q-2$-regular graph on atleast $2q-1$ vertices, in which case it is just the complete graph on $2q-1$ vertices. This recovers the result of \cite{raiD13}.

We can also choose $\widetilde{G}$ to be a biregular bipartite graph with $n_l$ vertices of degree $d_l$ each in one part and $n_r$ vertices of degree $d_r$ each in the other part. Without the source $s_\star$, the coding capacity for a network constructed from such a graph is
\begin{equation}
\frac{n_l+n_r}{n_l+n_r+\frac{n_ld_l+n_rd_r}{2}} = \frac{2}{2+\frac{n_ld_l+n_rd_r}{n_l+n_r}}= \frac{2}{2+d_{\text{av}}}
\end{equation}
where $d_{\text{av}}$ is the average degree of the graph. Here $d_{\text{av}}$ is not necessarily an integer and thus opens up more possibilities of coding capacities that can be achieved from this graph construction as opposed to \eqref{W1} where $d$ has to be an integer. %Also this upper bound can be achieved by following the achievability scheme outlined previously.
\begin{example}
Consider $\widetilde{G}$ to be the complete bipartite graph $K_{3,5}$. The sum-network $G$ constructed from this graph has coding capacity $\frac{2}{2+\frac{3 \times 5+5 \times 3}{8}} = \frac{8}{23}$
and has 23 sources and $24$ terminals. The approach of \cite{raiD13} would require $45 + \binom{45}{2}=1035$ sources and 1036 terminals.
\end{example}
It can also be seen that for any specified minimum cut between source terminal pairs, we can always choose an appropriate regular graph $\widetilde{G}$ so that the corresponding sum network has a capacity strictly smaller than one. Thus, it can be inferred that any min-cut ($= \alpha \times b$) between each source-terminal pair can never guarantee solvability.
%Another point of note is that the achievability scheme holds true for a class of connected $\widetilde{G}$. Thus, it can be inferred that any min-cut ($= \alpha \times b$) between each source-terminal pair can never guarantee solvability.


\section{Conclusions and future work}
\label{sec:conclusions}
In this paper, we have constructed a large class of sum-networks for which we can determine the capacity. These sum-networks are in general, smaller (with fewer sources and terminals) than sum-networks known to achieve the said capacity and answer a question raised in prior work. The construction of these sum-networks with the help of undirected graphs allows us to identify certain graph-theoretic properties that aid in constructing a capacity-achieving linear network code. Future work will involve analyzing these properties in greater detail and also examining whether there are other combinatorial structures that can be used to construct sum-networks. %In particular we would look to identify graph classes that ensure a feasible solution to \ref{opt-2}. Generalizations of the current construction using hypergraphs instead of simple graphs are also being pursued. %In addition, we have shown that unlike the multicast scenario, cut-based conditions cannot ensure solvability for sum-networks and hence by the equivalence in \cite{raiD12}, cut-based conditions cannot guarantee solvability in multiple unicast either.
\section*{Appendix}
\noindent {\it Proof of Claim \ref{claim:shortest_cycle}}.\\
We only need to rule out the possibility that $(i,j) \notin E_{\widetilde{Cyc}}$, but $(i,j) \in \widetilde{E}$. But this is impossible, since if $(i,j) \in \widetilde{E}$, we could find a shorter cycle in $\widetilde{G}$ than $\widetilde{Cyc}$.

\noindent {\it Proof of Claim \ref{claim:feas}}.\\
Suppose $\widetilde{G}$ is a $k$-regular graph on $b$ vertices. %Consider the cases when $b$ is odd and even.
If $b$ is even, then assigning
\begin{equation}
m_{(i,j)}(i)= \frac{b}{2} \text{~and~} m_{(i,j)}(j)= \frac{b}{2}~~ \forall (i,j) \in \widetilde{E}
\end{equation}
satisfies the first two constraints. Also, it satisfies the third constraint as
\begin{equation}
\sum_{j : (i,j) \in \text{In}_{\widetilde{G}}(i)}m_{(i,j)}(i) = \frac{kb}{2}=  |\widetilde{E}| ~~ \forall i \in \widetilde{V}.
\end{equation}

If $b$ is odd, then $k$ is even and we can construct an {\it Euler tour} \cite{diestel} of $\widetilde{G}$, i.e., a cycle that traverses every edge in $\widetilde{E}$ exactly once, though it may visit a vertex any number of times.
Suppose we start at vertex $1$ and the Euler tour is of the form
\begin{equation}
\left[(1,j_1),(j_1,j_2),\hdots , (j_f,1),(1,j_{f+1}), \hdots , (j_F,u)\right]
\label{et}
\end{equation}
where $f,f+1$ are such that the Euler tour passes through the sequence of vertices $j_f \rightarrow 1 \rightarrow j_{f+1}$ and $F$ is such that the edge $(j_F,u)$ is the last edge traversed on the Euler tour.

Then we can arrange the variables $m_{(i,j)}(i) ~\forall (i,j)\in \widetilde{E}$ as a one-to-one correspondence with \eqref{et} in the following way
\begin{align}\label{et2}
& [m_{(1,j_1)}(1),m_{(1,j_1)}(j_1),m_{(j_1,j_2)}(j_1), \hdots ,\\ \nonumber
&  m_{(j_f,1)}(1), m_{(1,j_{f+1})}(1), \hdots, m_{(j_F,1)}(1) ]
\end{align}

Assigning consecutive terms in \eqref{et2} alternately as $\left \lfloor \frac{b}{2} \right \rfloor$ and $\left \lceil \frac{b}{2} \right \rceil$, we see that for every vertex $i$, the value assigned to $m_{(i_{\text{in}},i)}(i)$ is different from $m_{(i,i_{\text{out}})}(i)$, where the Euler tour progresses in the order $i_{\text{in}} \rightarrow i \rightarrow i_{\text{out}}$. We can see that such an assignment satisfies the first constraint.

Also since there are only two possible values, we conclude that there are equal number of variables $m_{(i,j)}(i)$ with values equal to $\left \lfloor \frac{b}{2} \right \rfloor$ and $\left \lceil \frac{b}{2} \right \rceil$ for all $i \in \widetilde{V}$. Hence, we have $\forall i \in \widetilde{V}$
\begin{align*}
& \sum_{j : (i,j) \in \text{In}_{\widetilde{G}}(i)}m_{(i,j)}(i)= \\
& \frac{k}{2}\left \lfloor \frac{b}{2} \right \rfloor + \frac{k}{2}\left \lceil \frac{b}{2} \right \rceil = \frac{kb}{2} = |\widetilde{E}|
\end{align*}

%We demonstrate below that any graph $\widetilde{G}$ (not a tree) which satisfies the above condition has a corresponding sum-network $G$ via the construction discussed previously with $|\widetilde{V}|+|\widetilde{E}|+1$ sources and terminals such that there exists a linear network coding scheme that achieves the upper bound presented above \ref{UB1}.

%\subsection{Biregular bipartite graph}
Next, consider a biregular bipartite simple graph $\widetilde{G}$ with $n_l$ vertices in one part and $n_r$ vertices in the other part. Each vertex in the first part has degree $d_l$ and each vertex in the other part has degree $d_r$. Since $\widetilde{G}$ is simple, we must have that $d_l \leq n_r$ and $d_r \leq n_l$. Every edge $e$ in the graph is of the form $(i_l,i_r)$ where $i_l$ is a vertex in one part while $i_r$ is a vertex in the other part. Then it is easy to see that the following assignment satisfies the constraints in \ref{opt-1}.
\begin{align*}
m_e(i_l)&=n_l, \\
m_e(i_r)&=n_r,
\end{align*}
as $\forall l,r \in \widetilde{V}$ and $\forall e \in \widetilde{E}$
\begin{align*}
m_e(l)+m_e(r) = n_l + n_r = b, \text{and} \\
\sum_{e_i : e_i \in \text{In}(l)}m_{e_i}(l)= d_l \times n_l = |\widetilde{E}|.
\end{align*}

\noindent{\it Explanation regarding Remark 2}.\\
We evaluate, over two different finite fields, the computation capacity of an example sum-network shown in Figure \ref{fig}. The structure of the sum-network is better understood with the help of the following four subsets of the source nodes.
\begin{IEEEeqnarray*}{Rl}
A_1 =&\{s_1,s_{(1,2)},s_{(1,3)},s_{(1,4)},s_\star\}, \IEEEyesnumber \label{eq:sets}\\
A_2 =& \{s_2,s_{(1,2)},s_{(2,3)}, s_\star\},\\
A_3 =& \{s_3,s_{(1,3)},s_{(2,3)},s_{(3,4)}, s_\star\},\\
A_4 =&\{s_4,s_{(1,4)},s_{(3,4)}, s_\star\}.
\end{IEEEeqnarray*}
This sum-network is a modification of the sum-network in Figure \ref{fig:sum_network_K3*}. Unlike Figure \ref{fig}, there is no edge $(s_\star, \text{tail}(e_4))$ in Figure \ref{fig:sum_network_K3*}. The inclusion of this edge is the only difference between those two sum-networks.
\begin{figure}
  \centering
  \includegraphics[scale = 0.53]{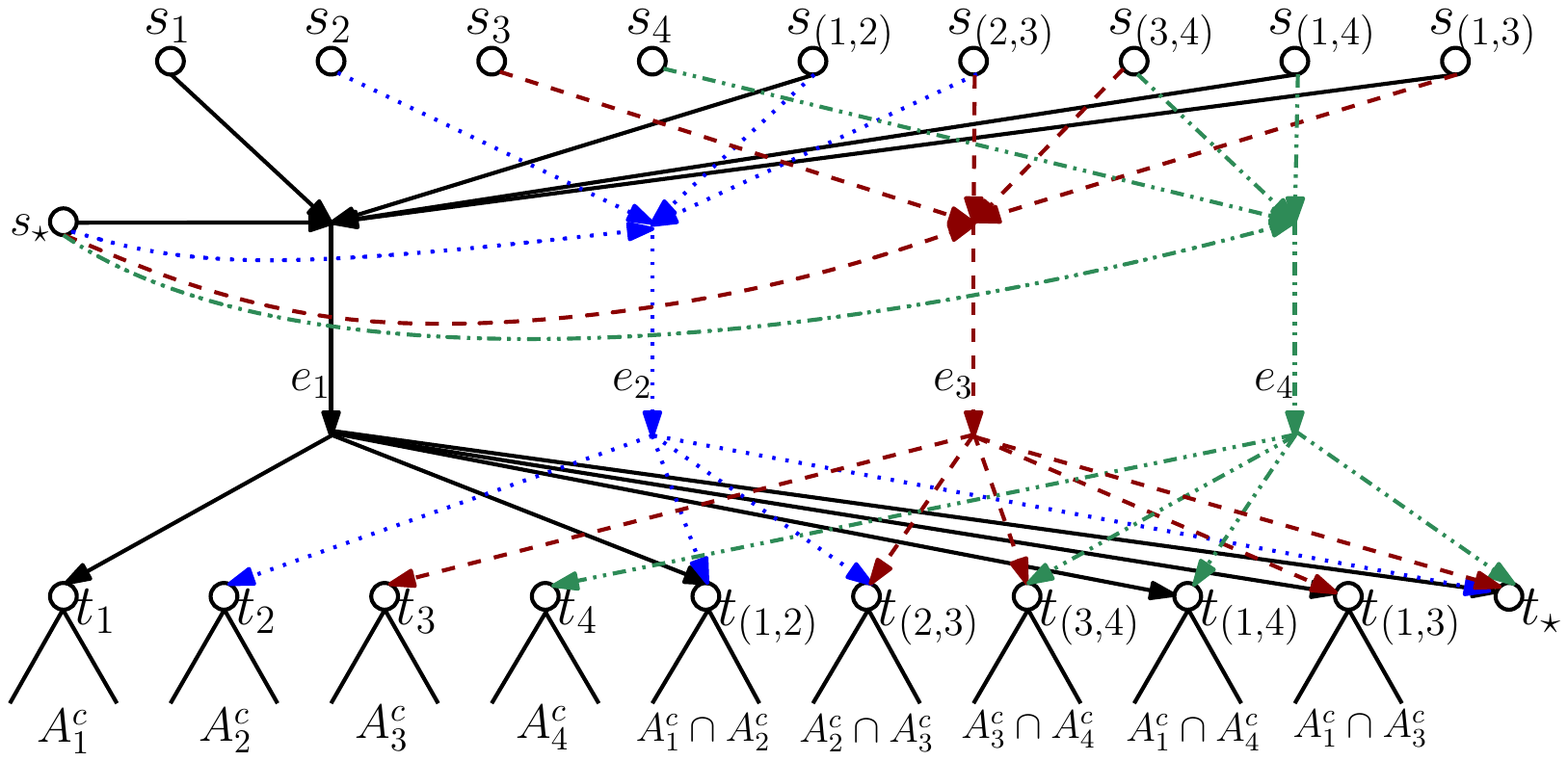}
  \caption{An example sum-network with $10$ sources and $10$ terminals. This sum-network is obtained by including an extra edge $(s_\star, \text{tail}(e_4))$ in Figure \ref{fig:sum_network_K3*}. We describe the direct edges to each terminal with the help of the sets defined in Eq. \eqref{eq:sets}. The direct edges to each terminal are indicated by the set shown below it, where superscript $c$ denotes the complement of the set.}\label{fig}
\end{figure}
Let $\mathcal{A}$ be the finite field alphabet which is used for communication.% Then we have the following.
\begin{claim}
  Suppose there is a $(m,n)$-network code that allows each terminal in the example sum-network to compute the finite field sum over $\mathcal{A}$ of all the source messages. Let ch$(\mathcal{A})$ denote the characteristic of the finite field $\mathcal{A}$. Then
  \begin{itemize}
    \item if ch$(\mathcal{A})=2, m/n \leq 4/9$, and
    \item if ch$(\mathcal{A})\neq 2, m/n \leq 4/10$.
  \end{itemize}
\end{claim}
\begin{proof}
Suppose there is a $(m,n)$-network code that allows each terminal to compute
\begin{IEEEeqnarray*}{Rl}
  \Sigma :=& X_1 + X_2 + X_3 + X_4 + X_{(1,2)} + X_{(1,3)}\\ &{}+ X_{(1,4)}+ X_{(2,3)}+X_{(3,4)} + X_\star,
\end{IEEEeqnarray*} where the addition is over $\mathcal{A}$. Then by subtracting the values of the source messages obtained over the direct edges, each terminal of the form $t_i, i \in \{1,2,3,4\}$ can obtain the sum of a subset of the source messages from the information transmitted over the edge $e_i$. The particular \textit{partial} sums recovered are listed in Table \ref{tab:t_point}.
\begin{table}
\caption{The partial sums obtained by certain terminals after they subtract the message values received over the direct edges.}\label{tab:t_point}
\renewcommand\arraystretch{1.5}
  \centering
  \begin{tabular}{cc}
    \hline
    Terminal & Partial sum \\ \hline
    $t_1$ & $X_1+X_{(1,2)}+X_{(1,3)}+X_{(1,4)}+X_\star$ \\
    $t_2$ & $X_2+X_{(1,2)}+X_{(2,3)}+X_\star$ \\
    $t_3$ & $X_3+X_{(1,3)}+X_{(2,3)}+X_{(3,4)}+X_\star$ \\
    $t_4$ & $X_4+X_{(1,4)}+X_{(3,4)}+X_\star$ \\
    \hline
  \end{tabular}
\end{table}
Let the value of the partial sum computed by the terminal $t_i, i \in \{1,2,3,4\}$ as shown in Table \ref{tab:t_point} be denoted as $P_i$. Since the terminal $t_{(i,j)}$ for any $(i,j) \in \{(1,2),(1,3),(1,4),(2,3),(3,4)\}$ receives information from both the edges $e_i$ and $e_j$, it can carry out the following operation and obtain a partial sum.
\begin{equation}\label{eq:t_edge}
  P_i + P_j - \Sigma = X_{(i,j)} + X_\star.
\end{equation}
Let us consider the terminal $t_\star$. It receives information from all the four bottleneck edges. Hence it can compute all the partial sums listed in Table \ref{tab:t_point}. It can also compute all partial sums of the form shown in Eq. \eqref{eq:t_edge}. Note that since each $X_i$ and $X_{(i,j)}$ are uniform i.i.d. over $\mathcal{A}^m$, all the partial sums obtained above are linearly independent. However, by the network code used, the maximum number of distinct values that can be transmitted across the four bottleneck edges is $(|\mathcal{A}|^n)^4$. Hence we have that
\begin{equation*}%\label{ub_all}
  |\mathcal{A}|^{4n} \geq |\mathcal{A}|^{9m} \implies m/n \leq 4/9.
\end{equation*}
The above holds for any finite field $\mathcal{A}$. Now suppose ch$(\mathcal{A}) \neq 2$. Then we can obtain a tighter bound on the ratio $m/n$. In particular, terminal $t_\star$ can carry out the following operation based on the values received over $\{e_1, e_2, e_3, e_4\}$.
\begin{IEEEeqnarray*}{Rl}
  T =& P_1 + P_2 + P_3 + P_4 - (X_{(1,2)}+X_\star) - (X_{(1,3)}+X_\star) \\ &{}- (X_{(1,4)}+X_\star) - (X_{(2,3)}+X_\star) - (X_{(3,4)}+X_\star),\\
   =& X_1 + X_2 + X_3 + X_4 +X_{(1,2)}+X_{(1,3)}+X_{(1,4)}\\&{}+X_{(2,3)}+X_{(3,4)}-X_\star.
\end{IEEEeqnarray*}
Then $t_\star$ can recover the value of $X_\star$ by the operation
\begin{equation*}
  2^{-1}(\Sigma - T),
\end{equation*}where the inverse exists as $2 \neq 0$ in a finite field $\mathcal{A}$ which has ch$(\mathcal{A})\neq 2$. Thus, terminal $t_\star$ can obtain the values of ten linearly independent values in $\mathcal{A}^m$ from the information received over the bottleneck edges. Hence,
\begin{equation*}
  |\mathcal{A}|^{4n} \geq |\mathcal{A}|^{10m} \implies m/n \leq 4/10.\qedhere
\end{equation*}
\end{proof}
\subsection{Linear network codes with rate equal to the upper bound}
\begin{table*}
\caption{The function values (each in $\mathcal{A}$, with ch$(\mathcal{A})=2$) transmitted across $e_1, e_2, e_3, e_4$ in Figure \ref{fig}. This network code has rate $=4/9$. Each message $X_1, X_2, X_3, X_4, X_{(1,2)},X_{(1,3)},X_{(1,4)},X_{(2,3)},X_{(3,4)}, X_{\star}$ is a vector with $4$ components, and $\phi_1(X),\phi_2(X), \phi_3(X), \phi_4(X)$ are vectors with $9$ components each. The number inside square brackets adjoining a vector indicates a particular component of the vector. Each terminal $t_{(i,j)}$ for any $(i,j) \in \{(1,2),(1,3),(1,4),(2,3),(3,4)\}$ can recover the value of $X_{(i,j)}+X_\star$ from this network code, which is then used in computing the sum over $\mathcal{A}$.}
\label{tab:K3*_code}
\renewcommand\arraystretch{1.5}
\centering
\begin{tabular}{ccccc}
\hline\hline
Component & $\phi_1(X)$ & $\phi_2(X)$ & $\phi_3(X)$ & $\phi_4(X)$\\
\hline
$1$ to $4$ & $\sum_{s_\alpha \in A_1}X_\alpha$ & $\sum_{s_\alpha \in A_2}X_\alpha$ & $\sum_{\alpha \in A_3}X_e$ & $\sum_{s_\alpha\in A_4}X_e$\\
$5$ & $X_{(1,3)}[1]+X_\star[1]$ & $X_{(1,2)}[1]+X_\star[1]$ & $X_{1,3}[3]+X_\star[3]$ & $X_{(1,4)}[4]+X_\star[4]$\\
$6$ & $X_{(1,3)}[2]+X_\star[2]$ & $X_{(1,2)}[2]+X_\star[2]$ & $X_{1,3}[4]+X_\star[4]$ & $X_{(3,4)}[1]+X_\star[1]$\\
$7$ & $X_{(1,4)}[1]+X_\star[1]$ & $X_{(1,2)}[3]+X_\star[3]$ & $X_{2,3}[2]+X_\star[2]$ & $X_{(3,4)}[2]+X_\star[2]$\\
$8$ & $X_{(1,4)}[2]+X_\star[2]$ & $X_{(1,2)}[4]+X_\star[4]$ & $X_{2,3}[3]+X_\star[3]$ & $X_{(3,4)}[3]+X_\star[3]$\\
$9$ & $X_{(1,4)}[3]+X_\star[3]$ & $X_{(2,3)}[1]+X_\star[1]$ & $X_{2,3}[4]+X_\star[4]$ & $X_{(3,4)}[4]+X_\star[4]$\\
[1ex]
\hline
\end{tabular}
\end{table*}
Table \ref{tab:K3*_code} describes a linear network code for the example sum-network which has rate $=4/9$ if ch$(\mathcal{A})=2$.  Note that each terminal $t_{(i,j)}$ can recover the value of $X_{(i,j)}+X_\star$ from this network code. For instance, based on Table \ref{tab:K3*_code}, terminal $t_{(1,3)}$ obtains $X_{(1,3)}[1]+X_\star[1], X_{(1,3)}[2]+X_\star[2]$ from the $5$th, $6$th components of $\phi_1(X)$ respectively; and it obtains $X_{(1,3)}[3]+X_\star[3], X_{(1,3)}[4]+X_\star[4]$ from the $5$th, $6$th components of $\phi_3(X)$ respectively. The decoding function for the terminals $t_i, i \in \{1,2,3,4\}$ is immediate. For a terminal of the form $t_{(i,j)}$, its decoding procedure involves computing the following partial sum.
\begin{equation*}
\sum_{s_\alpha \in A_i} X_\alpha + \sum_{s_\alpha \in A_j} X_\alpha - (X_{(i,j)}+X_\star) = \sum_{s_\alpha \in A_i \cup A_j}X_\alpha.
\end{equation*}  The source messages not present in the above partial sum are available to $t_{(i,j)}$ through the direct edges, and hence it can compute the required sum.
We use the fact that ch$(\mathcal{A})=2$ in the decoding procedure for terminal $t_\star$. Specifically, it can carry out the operation in Eq. \ref{eq:t_star_decode}.
%\newpage
\begin{IEEEeqnarray*}{L}
\sum_{i=1}^{4}\sum_{s_\alpha \in A_i}X_\alpha + \left(X_{(1,2)}+X_\star\right) + \left(X_{(1,3)}+X_\star\right)\\  + \left(X_{(1,4)}+X_\star\right) + \left(X_{(2,3)}+X_\star\right) + \left(X_{(3,4)}+X_\star\right)\IEEEyesnumber\label{eq:t_star_decode}\\
=\sum_{i=1}^{4}X_i  + 3\left(X_{(1,2)}+X_{(1,3)}+X_{(1,4)}+X_{(2,3)}+X_{(3,4)}\right)\\
\qquad+ 9X_\star, \\
= \Sigma,
\end{IEEEeqnarray*} as $9 \equiv 3 \equiv 1 \mod \text{ch}(\mathcal{A})$.
%\newpage
\begin{table*}
\caption{The function values (each in $\mathcal{A}$, with ch$(\mathcal{A})\neq 2$) transmitted across $e_1, e_2, e_3, e_4$ in Figure \ref{fig}. This network code has rate $=4/10$. Each message $X_1, X_2, X_3, X_4, X_{(1,2)},X_{(1,3)},X_{(1,4)},X_{(2,3)},X_{(3,4)}, X_{\star}$ is a vector with $4$ components, and $\phi_1(X),\phi_2(X), \phi_3(X), \phi_4(X)$ are vectors with $10$ components each. The number inside square brackets adjoining a vector indicates a particular component of the vector. Each terminal $t_{(i,j)}$ for any $i,j \in \{(1,2),(1,3),(1,4),(2,3),(3,4)\}$ can recover the value of $X_{(i,j)}+X_\star$ from this network code, which is then used in computing the sum over $\mathcal{A}$.}
\label{tab:nw_code}
\renewcommand\arraystretch{1.5}
\centering
\begin{tabular}{ccccc}
\hline\hline
Component & $\phi_1(X)$ & $\phi_2(X)$ & $\phi_3(X)$ & $\phi_4(X)$\\
\hline
$1$ to $4$ & $\sum_{s_\alpha \in A_1}X_\alpha$ & $\sum_{s_\alpha \in A_2}X_\alpha$ & $\sum_{\alpha \in A_3}X_e$ & $\sum_{s_\alpha\in A_4}X_e$\\
$5$ & $X_{(1,3)}[1]+X_\star[1]$ & $X_{(1,2)}[1]+X_\star[1]$ & $X_{1,3}[3]+X_\star[3]$ & $X_{(1,4)}[4]+X_\star[4]$\\
$6$ & $X_{(1,3)}[2]+X_\star[2]$ & $X_{(1,2)}[2]+X_\star[2]$ & $X_{1,3}[4]+X_\star[4]$ & $X_{(3,4)}[1]+X_\star[1]$\\
$7$ & $X_{(1,4)}[1]+X_\star[1]$ & $X_{(1,2)}[3]+X_\star[3]$ & $X_{2,3}[2]+X_\star[2]$ & $X_{(3,4)}[2]+X_\star[2]$\\
$8$ & $X_{(1,4)}[2]+X_\star[2]$ & $X_{(1,2)}[4]+X_\star[4]$ & $X_{2,3}[3]+X_\star[3]$ & $X_{(3,4)}[3]+X_\star[3]$\\
$9$ & $X_{(1,4)}[3]+X_\star[3]$ & $X_{(2,3)}[1]+X_\star[1]$ & $X_{2,3}[4]+X_\star[4]$ & $X_{(3,4)}[4]+X_\star[4]$\\
$10$ & $X_\star[1]$ & $X_\star[2]$ & $X_\star[3]$ & $X_\star[4]$\\
[1ex]
\hline
\end{tabular}
\end{table*}

If ch$(\mathcal{A}) \neq 2$, then the upper bound is $4/10$ and the network code used is as shown in Table \ref{tab:nw_code}. Note that this network code has the same values for its first nine components as in the previous network code in Table \ref{tab:K3*_code}. The decoding for all terminals except $t_\star$ is the same as in the previous case. However the decoding procedure for terminal $t_\star$ has to be different as the operation done in Eq. \ref{eq:t_star_decode} does not return the value $\Sigma$ as both coefficients $3,9$ are not equal to $1$ when ch$(\mathcal{A})\neq 2$. From Table \ref{tab:nw_code}, the $10$th component transmitted along each bottleneck edge is a distinct component of the source message $X_\star$. Since the terminal $t_\star$ receives information from each bottleneck edge, it can thus obtain the value of $X_\star$. Since $t_\star$ can also find the value of $X_{(i,j)}+X_\star$ based on the network code, it can recover the value of the message $X_{(i,j)}$. Thus, all the summands in $\sum_{s_\alpha \in A_i}X_\alpha$ except $X_i$ are known to $t_\star$ for all $i\in \{1,2,3,4\}$; moreover, it knows the value of the sum $\sum_{s_\alpha \in A_i}X_\alpha$ from the first four components of the network code. Thus $t_\star$ can obtain the value of each source message using the network code and thus can compute the required sum.


%\bibliographystyle{IEEEtran}
%\begin{thebibliography}{1}
%\end{thebibliography}
\bibliographystyle{IEEEtran}
\bibliography{tip,career_bib}

\begin{thebibliography}{10}
\providecommand{\url}[1]{#1}
\csname url@samestyle\endcsname
\providecommand{\newblock}{\relax}
\providecommand{\bibinfo}[2]{#2}
\providecommand{\BIBentrySTDinterwordspacing}{\spaceskip=0pt\relax}
\providecommand{\BIBentryALTinterwordstretchfactor}{4}
\providecommand{\BIBentryALTinterwordspacing}{\spaceskip=\fontdimen2\font plus
\BIBentryALTinterwordstretchfactor\fontdimen3\font minus
  \fontdimen4\font\relax}
\providecommand{\BIBforeignlanguage}[2]{{%
\expandafter\ifx\csname l@#1\endcsname\relax
\typeout{** WARNING: IEEEtran.bst: No hyphenation pattern has been}%
\typeout{** loaded for the language `#1'. Using the pattern for}%
\typeout{** the default language instead.}%
\else
\language=\csname l@#1\endcsname
\fi
#2}}
\providecommand{\BIBdecl}{\relax}
\BIBdecl

\bibitem{appuFKZ11}
R.~Appuswamy, M.~Franceschetti, N.~Karamchandani, and K.~Zeger, ``Network
  coding for computing: Cut-set bounds,'' \emph{IEEE Trans. on Info. Th.},
  vol.~57, no.~2, pp. 1015--1030, Feb 2011.

\bibitem{appuFKZ13}
------, ``Linear codes, target function classes, and network computing
  capacity,'' \emph{IEEE Trans. on Info. Th.}, vol.~59, no.~9, pp. 5741--5753,
  Sept 2013.

\bibitem{ramamoorthyL13}
A.~Ramamoorthy and M.~Langberg, ``{Communicating the sum of sources over a
  network},'' \emph{IEEE J. Select. Areas Comm.}, vol. 31(4), pp. 655--665,
  2013.

\bibitem{raiD12}
B.~K. Rai and B.~K. Dey, ``On network coding for sum-networks,'' \emph{IEEE
  Trans. on Info. Th.}, vol.~58, no.~1, pp. 50 --63, 2012.

\bibitem{huangR14}
S.~Huang and A.~Ramamoorthy, ``{On the multiple unicast capacity of 3-source,
  3-terminal directed acyclic networks},'' \emph{IEEE/ACM Trans. on
  Networking}, vol. 22(1), pp. 285--299, 2014.

\bibitem{huangR12_TCOM}
------, ``{An achievable region for the double unicast problem based on a
  minimum cut analysis},'' \emph{IEEE Trans. on Comm.}, vol. 61(7), pp.
  2890--2899, 2013.

\bibitem{kornerM79}
J.~Korner and K.~Marton, ``{How to encode the modulo-2 sum of binary
  sources},'' \emph{IEEE Trans. on Info. Th.}, vol. 25, no. 2, pp. 219--221,
  1979.

\bibitem{orlitskyR01}
A.~Orlitsky and J.~R. Roche, ``{Coding for computing},'' \emph{IEEE Trans. on
  Info. Th.}, vol. 47, no. 3, pp. 903--917, 2001.

\bibitem{doshiSMJ07}
V.~Doshi, D.~Shah, M.~M\'{e}dard, and S.~Jaggi, ``{Distributed Functional
  Compression through Graph Coloring},'' in \emph{{Data Compression
  Conference}}, 2007, pp. 93--102.

\bibitem{ramamoorthy08}
A.~Ramamoorthy, ``{Communicating the sum of sources over a network},'' in
  \emph{IEEE Intl. Symposium on Info. Th.}, 2008, pp. 1646--1650.

\bibitem{SD10}
S.~Shenvi and B.~K. Dey, ``{A Necessary and Sufficient Condition for
  Solvability of a 3s/3t sum-network},'' in \emph{IEEE Intl. Symposium on Info.
  Th.}, 2010, pp. 1858--1862.

\bibitem{rai13capacity}
B.~K. Rai and N.~Das, ``On the capacity of ms/3t and 3s/nt sum-networks,'' in
  \emph{IEEE Information Theory Workshop (ITW)}, 2013, pp. 1--5.

\bibitem{raiD13}
------, ``On the capacity of sum-networks,'' in \emph{2013 51st Annual Allerton
  Conference on Communication, Control, and Computing}, 2013, pp. 1545--1552.

\bibitem{tripathyR15}
A.~Tripathy and A.~Ramamoorthy, ``Capacity of sum-networks for different
  message alphabets,'' in \emph{IEEE Intl. Symposium on Info. Th.}, June 2015,
  pp. 606--610.

\bibitem{diestel}
R.~Diestel, \emph{Graph Theory, 4th Edition}.\hskip 1em plus 0.5em minus
  0.4em\relax Springer, 2012.

\end{thebibliography}

\end{document}